\documentclass[final,5p,times,fleqn,twocolumn,numbers]{elsarticle}
\usepackage{amsthm} 
\usepackage[intlimits,fleqn]{amsmath}    
\usepackage[T1]{fontenc}
\usepackage[utf8]{inputenc}

\usepackage{graphicx}
\usepackage{textcomp}
\usepackage{epsfig}
\usepackage{bm}
\usepackage{amssymb}
\usepackage{url}
\usepackage{enumitem} 
\usepackage{multirow}
\usepackage{hhline}
\usepackage{booktabs}
\usepackage{mathtools}
\usepackage{makecell}
\usepackage[linesnumbered,boxed,commentsnumbered,ruled,vlined,longend]{algorithm2e}
\usepackage{comment}

\makeatother

\newtheorem{mydef}{Definition}

\newtheorem{asmp}{Assumption}

\newtheorem{myprs}{Proposition}

\newtheorem{problem}{Problem}



\usepackage{stackengine}

\newcommand{\bmat}[1]{\begin{bmatrix} #1 \end{bmatrix}}

\def\comment{\textcolor{red}}

\newcommand{\st}{{\rm s.t.}}

\newcommand{\m}{\boldsymbol}
\allowdisplaybreaks[4]
\usepackage[colorlinks = true,
linkcolor = blue,
urlcolor  = black,
citecolor = blue,
anchorcolor = blue]{hyperref}

\newcommand{\mc}[1]{\mathcal{#1}}
\newcommand{\mbb}[1]{\mathbb{#1}}
\newcommand{\mr}[1]{\mathrm{#1}}
\usepackage[framemethod=TikZ]{mdframed}
\mdfdefinestyle{MyFrame}{%
linecolor=black,
outerlinewidth=1.25pt,
roundcorner=1.25pt,
innerrightmargin=5pt,
innerleftmargin=5pt,}

\usepackage[noabbrev]{cleveref}

\DeclarePairedDelimiter\abs{\lvert}{\rvert}%
\DeclarePairedDelimiter\norm{\lVert}{\rVert}%

\makeatletter
\let\oldabs\abs
\def\abs{\@ifstar{\oldabs}{\oldabs*}}
\let\oldnorm\norm
\def\norm{\@ifstar{\oldnorm}{\oldnorm*}}
\makeatother

\usepackage[english]{babel}
\usepackage[super]{nth}

\usepackage{float}
\usepackage[caption = false]{subfig}

\usepackage{array}
\usepackage{threeparttable}

\SetKwRepeat{Do}{do}{while}%

\usepackage{lettrine}
\usepackage{balance}

\usepackage{titlesec}
\titlespacing*{\section}{2pt}{*0.7}{*0.5}
\titlespacing*{\subsection}{2pt}{*0.7}{*0.3}
\titlespacing*{\subsubsection}{2pt}{*0.7}{*0.3}
\AtBeginDocument{%
  \setlength{\abovedisplayskip}{3.7pt}%
  \setlength{\belowdisplayskip}{3.7pt}%
  \setlength{\abovedisplayshortskip}{3.5pt}%
  \setlength{\belowdisplayshortskip}{3.5pt}%
  \setlength{\textfloatsep}{2pt plus 2pt minus 2pt}%
  \setlength{\floatsep}{2pt plus 1pt minus 1pt}%
  \setlength{\intextsep}{2pt plus 1pt minus 1pt}%
  \setlength{\abovecaptionskip}{2pt}%
  \setlength{\belowcaptionskip}{2pt}%
}


\usepackage{tikz}
\usetikzlibrary{shapes.geometric, arrows.meta, decorations.pathmorphing, decorations.markings, calc, 3d}
\usetikzlibrary{positioning, fit, backgrounds}
\tikzset{
tangent/.style={decoration={
markings,mark=at position #1 with {
\coordinate (ta) at (0,0);
\coordinate (tb) at (0.1,0);
}
},postaction=decorate},
tangent/.default=0.5
}

\newcommand{\parstart}[1]{\noindent \textbf{#1.}\;}

\newcommand{\R}{\mathbb{R}}
\newcommand{\Rn}[1]{\mathbb{R}^{#1}}

\newcommand{\titlesc}[1]{\title{\Large \vspace{0.7cm} \LARGE \centering {\textsc{{#1}}}}}

\usepackage{caption}
\usepackage{etoolbox}


\AtBeginEnvironment{thebibliography}{%
  \small
  \setlength{\itemsep}{0pt}
  \setlength{\parskip}{-5pt}
  \setlength{\parsep}{-5pt}
  \renewcommand{\url}[1]{}%
}

\patchcmd{\thebibliography}{\small}{\small\setlength{\itemsep}{0pt}\setlength{\parskip}{0pt}}{}{}

\usepackage{lettrine}

\usepackage{amssymb}
\usepackage{amsthm}

\usepackage{balance} 

\def\comment#1{#1}

\usepackage{fancyhdr}
\newcommand{\InPressBanner}{\textsc{Advances in Water Resources, In Press, June 2026}}
\fancypagestyle{inpress}{%
  \fancyhf{}%
  \fancyhead[C]{\footnotesize\InPressBanner}%
  \fancyfoot[C]{\thepage}%
}
\fancypagestyle{pprintTitle}{%
  \fancyhf{}%
  \fancyhead[C]{\footnotesize\InPressBanner}%
  \fancyfoot[C]{\thepage}%
}
\AtBeginDocument{\pagestyle{inpress}}

\begin{document}

\newdimen\origiwspc%
\newdimen\origiwstr%
\origiwspc=\fontdimen2\font
\origiwstr=\fontdimen3\font

\begin{frontmatter}
\titlesc{Exploring Uncertainty Propagation in Coupled Hydrologic and Hydrodynamic Systems via Distribution-Agnostic State Space Analysis}

\affiliation[1]{organization={Vanderbilt University, Department of Civil and Environmental Engineering},
addressline={24th Avenue South}, 
city={Nashville}, 
postcode={37235}, 
state={TN},
country={USA}}

\author[1]{Mohamad H. Kazma\corref{cor1}}
\ead{mohamad.h.kazma@vanderbilt.edu}

\author[1]{Ahmad F. Taha}
\ead{ahmad.taha@vanderbilt.edu}

\cortext[cor1]{\scriptsize Corresponding author.}

\tnotetext[t1]{This work is supported by National Science Foundation under Grants 2152450.}

\begin{abstract}
Accurate overland runoff and infiltration predictions are critical for effective water resources management, in particular for urban flood management. However, the inherent uncertainty in rainfall patterns, soil properties, and initial conditions makes reliable flood forecasting a challenging task. This paper presents a framework for quantifying the impact of these uncertainties on hydrologic and hydrodynamic simulations via a state space approach based on a differential algebraic equation (DAE) formulation that couples surface and subsurface constraints with the governing dynamics. Under this formulation, the complex interactions between overland flow and infiltration dynamics are captured in realtime. To account for uncertainty in inputs and parameters, the proposed framework quantifies and propagates these uncertainties through the DAE model formulation under partial measurements. The effectiveness of the approach is demonstrated through a series of numerical experiments on synthetic and real world catchments, highlighting its ability to provide probabilistic estimates of watershed state conditions while accounting for uncertainty. An important aspect of the proposed methods is that they are distribution-agnostic, i.e., they only require covariances of uncertainty and not specific types of distributions. The proposed framework is further validated against Monte Carlo (MC) ensemble simulations while providing probabilistic state estimates for measured and unmeasured watershed states under partial gauging.
\end{abstract}

\begin{keyword}
Differential algebraic equations\sep uncertainty quantification\sep probabilistic state estimation\sep hydrologic-hydrodynamic modeling\sep watershed monitoring \vspace{-0.3cm}
\end{keyword}

\end{frontmatter}
\vspace{-0.3cm}

\section{Introduction and paper contributions}\label{sec:Into-Lit}
\vspace{0.1cm}
\lettrine{M}{ONITORING} of watersheds is critical towards the understanding of hydrological processes and the management of water resources, in particular predicting flood hazards and erosion control. Due to limited spatial monitoring and the nonlinear and time-varying nature of watershed response to rainfall event driven dynamics, hydrologic and hydrodynamic models are used for simulating and predicting watershed response to such events~\cite{Jiang2019}. The simulation of such models can \textit{(i)} be based on several coupling strategies such as internal or external coupling of surface and subsurface flow dynamics~\cite{Zhang2025g}; \textit{(ii)} mathematically represent complex shallow water formulations as dynamic wave, kinematic wave, or local inertia approximations~\cite{Kim2012a,Gomes2023,Cozzolino2019}; and \textit{(iii)} consider implementing either 1D or 2D simulations in EPA’s Storm Water Management (SWMM)~\cite{rossman2010storm}, HydroPol2D~\cite{Gomes2023}, Weighted Cellular Automata 2D (WCA2D)~\cite{Guidolin2016}, and the Hydrologic Engineering Center’s Hydrologic Modeling System and River Analysis System (HEC-HMS/RAS)~\cite{Artiglieri2025,Brunner2016}. Readers can refer to~\cite{Gomes2023} for the temporal and spatial scale that can be simulated in each of the models in the literature. The modeling choices enable computationally efficient simulation of the complex physics involved in a watershed hydrologic and hydrodynamic response. However, the accuracy of these models is often limited by uncertainty arising in rainfall forecasts, underlying land and soil properties, initial conditions, and prevalence of ungauged areas. This can degrade realtime forecast reliability and, in turn, limit operational decisions for the control and management of water resources. Reliable monitoring and prediction of watershed conditions under uncertainty therefore become critical for reliable forecasting of a watershed's hydrodynamic response used to inform decisions in water resources management.\vspace{-0.1cm}

Motivated by increasing forecast variability and imperfect knowledge of soil and land/surface parameters due to exacerbating climate conditions~\cite{Chen2015a}, the objective of this paper is to introduce a state space representation of a coupled hydrologic-hydrodynamic watershed model, while providing a systems theoretic framework for realtime probabilistic monitoring of watershed conditions. By representing the dynamics as a coupled set of differential and algebraic equations, we obtain a form amenable to estimation and uncertainty propagation using control theoretic methods available in the literature, thereby providing a framework to quantify how rainfall and land/soil parameter uncertainty impacts predicted watershed states and outputs under limited watershed gauging.

\vspace{-0.15cm}
\subsection{Relevant literature}\label{subsec:literature}
Watershed modeling literature, with the advent of increasing computational power and data availability, has moved toward more distributed and coupled representations of rainfall-runoff-infiltration processes over 2D terrains. Model selection depends on watershed scale, data availability, and operational objective~\cite{Chen2015a,Borah2011,Li2018b,Shen2023a,Zhang2025g}. In practice, 2D hydrologic-hydrodynamic models are often used for urban and small watershed flood applications, whereas 1D or semi distributed models are more common for larger basins~\cite{Kim2012a,Liu2019a,Shu2024,Gomes2025}; however, 2D formulations increase state and parameter dimensions and thus computational demand. Accordingly, the intended application, such as flood forecasting, water quality modeling, or ecosystem management, determines the required model complexity, coupling strategy, and validation level~\cite{DarenHarmel2007}. Furthermore, it is essential to account for uncertainty arising from \textit{(i)} meteorological rainfall forcing and \textit{(ii)} distributed watershed parameters. It is noted that a major source of discrepancy in watershed outputs, including runoff timing and peak magnitudes~\cite{Patil2012,McCuen2009,Melching1990}, is due to input uncertainty. Recent studies have investigated spatiotemporal rainfall uncertainty and its direct impact on flood hazard forecasts in semiarid and climate sensitive basins~\cite{Tsvetkova2019,Xiaojun2021,Tang2023,Gomes2025}. Parameter uncertainty in infiltration and surface roughness has also been shown to result in larger variability in watershed conditions~\cite{Gupta2019,Gupta2023,Wiebe2025}.

That being said, methods that quantify or propagate uncertainty in watershed hydrologic-hydrodynamic models include information theoretic methods, Bayesian inference, generalized likelihood uncertainty estimation, and Monte Carlo (MC) sampling~\cite{Beven2001,Gupta2023}. These methods have been applied to a range of hydrologic models, including lumped, semi distributed, and fully distributed formulations~\cite{Gupta2023}. For instance, a likelihood based approach is used in watershed water quality modeling because it links model simulations to uncertain observations, allowing the identification of behavioral parameter sets and the quantification of how input and parameter uncertainty propagate to watershed states and outputs, including flow and pollutant concentrations~\cite{Jia2008,Freni2009}. However, the aforementioned approach does not directly provide a realtime uncertainty map for measured and unmeasured states under poor and partial measurement conditions. Furthermore, Bayesian inference methods using Markov Chain MC sampling have been considered to estimate posterior distributions of model parameters and states given uncertain observations and inputs~\cite{Cibin2014}. Probabilistic flood mapping studies using MC simulations provide a practical approach for providing confidence aware flood forecasts under uncertain rainfall and soil parameters~\cite{Moges2020,Nguyen2024a,Gomes2025,Castro2025}; however, numerous model realizations are required to characterize uncertainty bounds, in particular for large distributed watershed models.

Several studies have shown that Kalman filter estimation~\cite{Kalman1960} and its variants can be used for realtime state estimation and forecasting under measurement uncertainty in hydrologic and hydrodynamic models. Ensemble Kalman smoother and multivariate ensemble Kalman filter (EnKF) studies show that uncertainty in hydrologic model parameters directly affects watershed condition estimates~\cite{Lei2014,Shi2015}. An unscented Kalman filter (UKF) for rainfall-runoff models is shown to improve state reconstruction and output trajectories in synthetic and real basins~\cite{Jiang2019}. In spatially distributed watershed models, a partitioned EnKF update scheme is proposed while accounting for high dimensional augmentation~\cite{Xie2013a}. A realtime EnKF based state parameter updating is also shown to improve predictions at ungauged locations when downstream observations are incorporated~\cite{Xie2014}. A stormwater digital twin with extended Kalman filtering and online sensor quality optimization is proposed for forecasting~\cite{Kim2025}. A physics guided sensor to model EnKF framework is developed for realtime soil moisture estimation and forecasting~\cite{Zhao2026}. However, from a systems theoretic perspective, reliable state estimation with EKF/EnKF requires a measurement configuration that renders the watershed model observable 
or at least detectable; otherwise, uncertainty mapping for unmeasured states can become unreliable under partial sensing. This requires introducing a framework for hydrologic and hydrodynamic modeling, that presents an explicit state space representation used to analyze observability and support reliable realtime monitoring.

On this basis, some studies have introduced state space models and sensor to model assimilation frameworks in stormwater systems. A state space 1D watershed reservoir channel model for reactive and predictive valve control is introduced in~\cite{GomesJunior2022}. A graph theoretic placement of control structures to systematically reduce hydrograph peaks at the watershed scale is developed in~\cite{Bartos2019}. A digital twin framework is formulated in~\cite{Bartos2021a}, representing the Saint Venant dynamics in state space form combined with extended Kalman filtering and model predictive control for real time operations~\cite{Bartos2021a}. The studies show the potential of considering state space modeling to improve the reliability of watershed monitoring.
\begin{figure*}[t]
    \centering
\includegraphics[keepaspectratio=true,width=\textwidth]{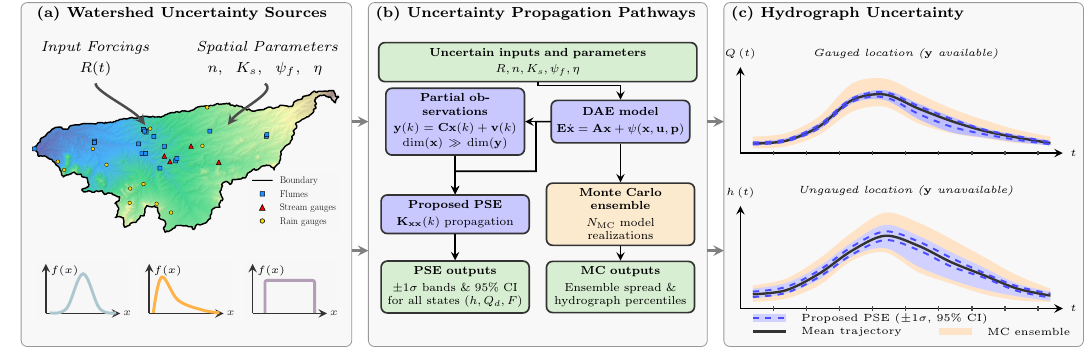}
\vspace{0.05cm}
    \caption{Conceptual framework representing the proposed realtime covariance propagation framework for 2D overland flow and infiltration watershed modeling under uncertainty; (a) illustrates the watershed sensing configuration with flumes, stream gauges, and rain gauges, together with the five uncertainty sources and their respective probabilistic distribution models; (b) shows the two propagation pathways: the proposed realtime covariance propagation framework under partial measurements and an MC ensemble pathway used as reference; and (c) compares resulting uncertainty intervals at gauged and ungauged locations, illustrating that the proposed framework provides tighter measurement informed confidence intervals at monitored locations while still providing uncertainty estimates at unmeasured locations.}\label{fig:framework_outcome_uncertainty}
    \vspace{-0.3cm}
\end{figure*}
\vspace{-0.15cm}
\subsection{Research gaps, novelty and contributions}\label{subsec:contributions}
The aforementioned literature provides essential advances in watershed modeling, uncertainty analysis, and realtime estimation; however, it is not void of limitations. While MC based uncertainty studies can provide confidence intervals on system conditions, the number of model realizations required limits scalability for realtime forecasting in distributed watershed models~\cite{Moges2020,Nguyen2024a,Gomes2025,Castro2025}. Kalman based methods improve state estimation and forecasting; however, reliable estimation also requires sufficient observability under the available measurements~\cite{Jiang2019,Xie2013a,Xie2014,Lei2014,Shi2015,Kim2025,Zhao2026}. Furthermore, the current systems oriented watershed literature does not provide a unified coupled 2D hydrologic-hydrodynamic state space framework~\cite{GomesJunior2022,Bartos2019,Bartos2021a}.
To that end, we introduce a coupled 2D state space representation of the hydrologic and hydrodynamic processes and develop a framework for realtime probabilistic watershed monitoring under partial watershed gauging. In this formulation, the underlying overland flow and infiltration processes are modeled as a nonlinear system of differential algebraic equations (DAEs).
The main contributions are as follows.
\begin{itemize}[leftmargin=*]
\item We introduce a coupled state space framework for 2D hydrologic and hydrodynamic dynamics, specifically overland flow routing and Green-Ampt infiltration, under a diffusive wave formulation. The framework preserves constrained flow interactions through algebraic relations while retaining the distributed spatiotemporal dynamics of overland flow and infiltration. Furthermore, we show that the resulting semi-explicit differential algebraic model is index one and regular. This implies that we obtain locally unique solutions for the system state variables under implicit discretization formulations and consistent initial conditions.\vspace{-0.25cm}
\item We develop a realtime probabilistic monitoring framework by embedding state estimation and uncertainty propagation within the proposed DAE model under partial measurements. Under this formulation, uncertainty in forcing and soil parameters is mapped to confidence aware estimates of watershed conditions in partially gauged settings. \comment{The covariance propagation of the proposed framework is distribution-agnostic;} it requires only the means and variances of uncertain inputs and parameters, without assuming a specific probability distribution. \comment{A distributional choice is introduced only at the confidence interval computation in~\eqref{eq:normal_propagation}, where the propagated covariance is mapped to confidence intervals using the adopted distribution model for each variable.} This is in contrast to MC and Bayesian methods that require sampling from a defined distribution. \vspace{-0.25cm}
\item A comprehensive case study is provided on benchmark and real watershed settings, including V-Tilted and Walnut Gulch Experimental Watershed (WGEW), to evaluate the proposed framework under varying rainfall and soil uncertainty scenarios with partial sensing and to demonstrate scalability from synthetic to real world distributed watershed conditions. MC uncertainty envelopes are generated using HydroPol2D cellular automata and local inertia formulations, and are used as references to assess estimation quality and computational scalability.
\end{itemize}

\parstart{Paper Notation} Let $\mathbb{N}$, $\R$, and $\R_{\geq 0}$ denote the set of natural, real, and nonnegative real numbers. Let $\Rn{n}$ and $\Rn{n\times m}$ denote the set of real-valued column vectors of size $n$, and $n$-by-$m$ real matrices. The cardinality of a set $\mc{S}$ is denoted by $|\mc{S}|$. For a matrix $\m A$, $\mr{det}(\m A)$, $\mr{trace}(\m A)$, and $\mr{rank}(\m A)$ denote determinant, trace, and rank, respectively. The Moore-Penrose pseudoinverse of $\m A$ is denoted by $\m A^\dagger$. The operator $\mr{sgn}(\cdot)$ denotes the sign function. The operator $\mr{diag}\left(\{\gamma_j\}_{j=1}^{n}\right)\in\Rn{n\times n}$ constructs a diagonal matrix from $\m\gamma=[\gamma_1,\ldots,\gamma_n]^\top\in\Rn{n}$. The stacked vector notation $\{\m x_i\}_{i=0}^{N}\in\Rn{(N+1)n}$ denotes column concatenation of vectors $\m x_i\in\Rn{n}$ for $i\in\{0,\ldots,N\}$. Variable $k\in\mbb{N}$ denotes the discrete-time index. The observation horizon is denoted by $N_T\in\mbb{N}$. Perturbations are denoted by $\delta(\cdot)$ and nominal trajectories by superscript $(\cdot)^0$. For random vectors $\m a,\m b$, the cross-covariance matrix is denoted by $\mr{K}_{\m a\m b}:=\mr{Cov}(\m a,\m b)$, and $\mr{K}_{\m a\m a}$ denotes the covariance matrix. For a random vector $\m x=[x_1,\ldots,x_n]^\top\in\Rn{n}$, define $\mbb{E}(\m x):=[\mbb{E}(x_1),\ldots,\mbb{E}(x_n)]^\top$, $\mr{Var}(\m x):=[\mr{Var}(x_1),\ldots,\mr{Var}(x_n)]^\top$, and $\mr{K}_{\m x\m x}$ as the covariance matrix of $\m x$. The standard deviation of a scalar random variable $x_j$ is denoted by $\sigma_{x_j}:=\sqrt{\mr{Var}(x_j)}$; for a state variable $j$, $\sigma_j(k):=\sqrt{[\mr{K}_{\m x\m x}(k)]_{jj}}$ denotes the propagated standard deviation at time step $k$. The shorthand covariance operation over a linear relation $\m A\m x=\m b$ is written as $\mr{Cov}(\m A\m x=\m b)$, i.e., $\mr{Cov}(\m A\m x,\m A\m x)=\mr{Cov}(\m b,\m b)$. 

\parstart{Paper Organization}
The paper is organized as follows. Section~\ref{sec:methods} introduces the hydrological and hydrodynamic model considered, presents a state space representation of the coupled dynamics, and describes the proposed framework for quantifying uncertainty in watershed conditions under uncertain inputs and parameters. Section~\ref{sec:study_areas} describes the study areas considered and the datasets used. Section~\ref{sec:results} reports the main results. Finally, Section~\ref{sec:summary} concludes the manuscript and discusses limitations and directions for future work.

\section{Model and methods}\label{sec:methods}
\subsection{Hydrologic-hydrodynamic model}\label{subsec:hydropol2d}
We consider the coupled hydrologic and hydrodynamic dynamics adopted in HydroPol2D~\cite{Gomes2023} over a 2D domain $\mc{D}\subset\mbb{R}^2$. The model couples rainfall and runoff generation, Green-Ampt infiltration, and directional Manning routing. The input maps of topography, land cover, and soil type define the distributed parameters used by these components. In this section, we introduce the governing equations, and in Section~\ref{subsec:dae} we introduce the coupled system rewritten as a set of differential algebraic equations. The original HydroPol2D explicit cellular routing and local inertia formulations are used in the case studies section as baselines for the proposed methods~\cite{Gomes2023,Gomes2024a}. The conservation of mass under the diffusive wave approximation is given by~\eqref{eq:mass_balance}-\eqref{eq:inflow_outflow_local}, and the Green-Ampt infiltration model is given by~\eqref{eq:green_ampt_capacity}-\eqref{eq:infiltration_rate}. The directional discharge relation based on Manning's equation is given by~\eqref{eq:manning}.

\subsubsection{Conservation of mass and momentum}\label{subsubsec:mass_balance}
The domain $\mathcal{D} \subset \mathbb{R}^2$ represents a spatial discretization of a surface into grid cells, a central cell with four neighbors, forming a graph $\mathcal{G} = (\mathcal{K}, \mathcal{E})$. The cell set is denoted as $\mc{K}=\{(i,j):1\le i\le N_x,\;1\le j\le N_y\}$ with cardinality $K=N_x\times N_y$. Let cell $(i,j)$ denote a generic grid location, then $\mc{N}_{i,j}$ is the set of its neighboring cells (von Neumann neighborhood), and $\mc{D}_{\mr{dir}}=\{\mr{L},\mr{R},\mr{U},\mr{D}\}$ are its four cardinal directions. Then, for each cell $(i,j)$ with area $A_{i,j}=\Delta x\Delta y$ [m$^2$], the local rate of change in surface water depth $h^{i,j}(t)$ [m] is written as
\begin{equation}\label{eq:mass_balance}
\dot{h}^{i,j}(t)=R^{i,j}(t)-f^{i,j}(t)+\frac{1}{A_{i,j}}\left(Q_{\mr{in}}^{i,j}(t)-Q_{\mr{out}}^{i,j}(t)\right),
\end{equation}
where $R^{i,j}(t)$ is the rainfall intensity [m s$^{-1}$], $f^{i,j}(t):=f^{i,j}\!\left(h^{i,j}(t),F^{i,j}(t)\right)$ is the infiltration rate [m s$^{-1}$] that depends on the local ponding depth $h^{i,j}(t)$ and cumulative infiltration depth state $F^{i,j}(t)$ [m] (see~\eqref{eq:green_ampt_capacity}-\eqref{eq:infiltration_rate}), and the inflow $Q_{\mr{in}}^{i,j}(t)$ and outflow rate $Q_{\mr{out}}^{i,j}(t)$ [m$^3$ s$^{-1}$] satisfy
\begin{equation}\label{eq:inflow_outflow_local}
Q_{\mr{in}}^{i,j}(t)=\sum_{(r,s)\in\mc{N}_{i,j}}Q_{(r,s)\rightarrow(i,j)}(t),\qquad
Q_{\mr{out}}^{i,j}(t)=\sum_{d\in\mc{D}_{\mr{dir}}}Q_{d}^{i,j}(t),
\end{equation}
where the inflow rate $Q_{\mr{in}}^{i,j}(t)$ is the sum of directional discharges entering cell $(i,j)$ from its von Neumann neighboring cells, while the outflow rate $Q_{\mr{out}}^{i,j}(t)$ is the sum of directional discharges leaving cell $(i,j)$ toward its neighbors along each cardinal direction. Equation~\eqref{eq:mass_balance} then states that water depth dynamics at each cell are governed by rainfall input, infiltration loss, and net lateral flux exchange with neighboring cells. A positive right hand side implies local ponding growth, whereas a negative right hand side implies recession. Note that evapotranspiration is typically included in this formulation~\cite{rossman2010storm}; however, in this manuscript we neglect evapotranspiration since the focus herein is on short duration wet weather events, where rainfall, infiltration, and lateral routing dominate the water balance~\cite{Gomes2023}. Infiltration of water into the soil is represented using the Green-Ampt model, which is introduced next.

\subsubsection{Green-Ampt infiltration model}\label{subsubsec:green_ampt}
Infiltration is modeled using the Green-Ampt model~\cite{Green1911Studies}. The model represents infiltration as a function of ponding depth and cumulative infiltration, and is considered a simplified approximation of Richards unsaturated flow dynamics~\cite{richards1931capillary}. The infiltration capacity at cell $(i,j)$ is thus given by
\begin{equation}\label{eq:green_ampt_capacity}
f_{\mr{cap}}^{i,j}(t)=K_s^{i,j}\left[1+\frac{\Delta\theta^{i,j}\left(\psi_f^{i,j}+h^{i,j}(t)\right)}{\max\left(F^{i,j}(t),F_{\min}\right)}\right],
\end{equation}
where $K_s^{i,j}$ is saturated hydraulic conductivity [m s$^{-1}$], $\psi_f^{i,j}$ is wetting front suction head [m], $\Delta\theta^{i,j}=\theta_s^{i,j}-\theta_i^{i,j}$ is the moisture deficit [-], $F_{\min}>0$ is a small regularization threshold [m]. The actual infiltration rate is limited by water availability accumulated over a time step $\Delta t$ and is given by
\begin{equation}\label{eq:infiltration_rate}
f^{i,j}(t):= \dot{F}^{i,j}(t) = \min\left(f_{\mr{cap}}^{i,j}(t),\;R^{i,j}(t)+\frac{h^{i,j}(t)}{\Delta t}\right).
\end{equation}
Accordingly, Equation~\eqref{eq:green_ampt_capacity} states that infiltration capacity is composed of \textit{(i)} the conductivity term $K_s^{i,j}$ and \textit{(ii)} a capillary term that depends on $\Delta\theta^{i,j}\left(\psi_f^{i,j}+h^{i,j}(t)\right)/\max\left(F^{i,j}(t),F_{\min}\right)$. For small $F^{i,j}(t)$, the capillary term is larger and the infiltration capacity is higher; as $F^{i,j}(t)$ increases, the capillary contribution decreases and $f_{\mr{cap}}^{i,j}(t)$ approaches $K_s^{i,j}$. Equation~\eqref{eq:infiltration_rate} enforces water availability by limiting infiltration to the minimum between the soil capacity and the available surface water rate $R^{i,j}(t)+h^{i,j}(t)/\Delta t$; this prevents unrealistic infiltration that exceeds local water availability. \comment{We note that the Green-Ampt model used in~\eqref{eq:green_ampt_capacity}-\eqref{eq:infiltration_rate} represents infiltration through a sharp wetting front advancing in a homogeneous soil column, parameterized by $K_s^{i,j}$, $\psi_f^{i,j}$, and $\Delta\theta^{i,j}$. Several hydrologic behaviors are therefore not captured by this representation, including layered soil profiles, macropore and preferential flow, wetting and drying hysteresis, soil crusting and time-varying surface conductivity~\cite{Becker2018}, and lateral subsurface flow in the unsaturated zone~\cite{richards1931capillary,Gomes2023,Gomes2024a}.}

\subsubsection{Manning's equation and outflow routing}\label{subsubsec:manning}
The directional outflow from each cell is obtained from Manning's relation combined with cellular automata directional weights. For cell $(i,j)$ and direction $d\in\mc{D}_{\mr{dir}}$,
\begin{equation}\label{eq:manning}
Q_d^{i,j}(t)=\frac{w_d^{i,j}(t)}{n^{i,j}}\,\bigl(h^{i,j}(t)\bigr)^{5/3}\left|\frac{\Delta\eta_d^{i,j}(t)}{\Delta s_d}\right|^{1/2}\mr{sgn}\bigl(\Delta\eta_d^{i,j}(t)\bigr)\,\ell_d,
\end{equation}
where $\Delta\eta_d^{i,j}=(z^{i,j}+h^{i,j})-(z_d^{i,j}+h_d^{i,j})$ is the water surface elevation difference [m], $z^{i,j}$ is the bed elevation at cell $(i,j)$ [m], $z_d^{i,j}$ is the bed elevation of the neighboring cell in direction $d$ [m], $h_d^{i,j}$ is the neighboring water depth in direction $d$ [m], $n^{i,j}$ is Manning roughness coefficient [s m$^{-1/3}$]~\cite{chowv}, $\Delta s_d$ is the directional distance [m], and $\ell_d$ is the interface width [m] ($\Delta y$ for left/right and $\Delta x$ for up/down). The weights satisfy $w_d^{i,j}\in[0,1]$ and $\sum_{d\in\mc{D}_{\mr{dir}}}w_d^{i,j}=1$, allowing the distribution of flow to multiple downslope neighbors based on the weights~\cite{Guidolin2016}. 

Following HydroPol2D's original implementation~\cite{Gomes2023}, the nonlinear Manning relation is evaluated at each time step for the steepest water surface direction and then redistributed to the remaining admissible directions using cellular automata weights based on available neighbor void volumes~\cite{Guidolin2016}. In this hydrodynamic model, the diffusive wave assumption is adopted and the friction slope is approximated by the water surface (energy grade) slope~\cite{vieira1983conditions,Cozzolino2019}. While local inertia approximations can be more accurate in some flood simulation settings, their applicability is restricted by limiting conditions in certain flooding regimes and topographic settings~\cite{Cozzolino2019}. Accordingly, directional discharges are computed from local water surface gradients and Manning resistance, while the local and convective acceleration terms of the full shallow water momentum equation are neglected. \comment{Note that the diffusive wave approximation is consistent with subcritical regimes ($\mr{Fr}=u/\sqrt{gh}<1$) under mild to moderate slopes and slowly varying overland flow, where local and convective acceleration terms in the shallow water momentum equation are small relative to the gravity and friction terms~\cite{vieira1983conditions,Cozzolino2019}. These conditions are consistent with the topographic and runoff regimes of the ephemeral channels of Walnut Gulch considered in Section~\ref{subsec:results_solver}; the close agreement between the DAE outlet hydrograph and the HydroPol2D local inertia baseline in the case studies provides empirical evidence that inertial terms are small in these settings.}
\subsection{The dynamics in a state space representation}\label{subsec:dae}
We formulate the coupled hydrologic-hydrodynamic model, encompassing overland flow and Green-Ampt infiltration dynamics, as a semi-explicit nonlinear DAE. For a catchment domain $\mc{D} \subset \mbb{R}^2$, represented by the structured cell graph $\mc{G} = (\mc{K},\mc{E})$ with cell set $\mc{K}$, a directed edge $\big((i,j),(r,s)\big) \in \mc{E} \subseteq \mc{K} \times \mc{K}$ denotes an admissible flow connection from the neighboring cell $(i,j)$ to $(r,s)$. The boundary, outlet, and interior cell sets are $\mc{K}_{\partial}\subset\mc{K}$, $\mc{K}_{\mr{out}}\subseteq\mc{K}_{\partial}$, and $\mc{K}_{\mr{int}}=\mc{K}\setminus\mc{K}_{\partial}$, respectively. Each cell $(i,j)\in\mc{K}$ is characterized by area $A_{i,j}:=\Delta x\,\Delta y$ [m$^2$], elevation $z_{i,j}$ [m], and Manning roughness $n_{i,j}$ [$\mr{s\,m^{-1/3}}$].
The coupled hydrologic-hydrodynamic dynamics over a 2D catchment domain can be rewritten as a nonlinear semi-explicit DAE as follows
\begin{subequations}\label{eq:semi_NDAE_rep}
\begin{align}
\textit{mass \& infiltration dynamics}:\;\; \dot{\m x}_{d} &= \m{f}(\m x_d, \m x_a, \m u), \label{eq:X_d}\\
\textit{algebraic flow constraints}:\;\; \m 0 &= \m{g}(\m x_d, \m x_a), \label{eq:X_a}
\end{align}
\end{subequations}
where $\m x_d:=\m x_d(t)\in\Rn{n_d}$ is the differential state vector, $\m x_a:=\m x_a(t)\in\Rn{n_a}$ is the algebraic state vector, and $\m u:=\m u(t)\in\Rn{n_u}$ is the input vector, with $n_x=n_d+n_a$. The mapping $\m f(\cdot):\Rn{n_d}\times\Rn{n_a}\times\Rn{n_u}\rightarrow\Rn{n_d}$ describes the mass and infiltration dynamics under rainfall forcing, coupled through the Manning based cellular automata flow variables in $\m x_a$. The mapping $\m g(\cdot):\Rn{n_d}\times\Rn{n_a}\rightarrow\Rn{n_a}$ enforces the algebraic flow constraints from the directional Manning cellular automata map. We note that, under the above representation, the hydrologic and hydrodynamic states are coupled through algebraic constraints, meaning that the solution must satisfy both the mass and infiltration dynamics and the flow constraints simultaneously at each time step. The explicit state construction and state space DAE form presented in the subsequent sections reformulate the original HydroPol2D model, which uses an explicit cellular automata implementation of the mass balance and Manning routing equations. Under this reformulation, the DAE model captures these interactions through algebraic relations while retaining the distributed spatiotemporal dynamics of overland flow and infiltration. This enables a framework amenable for state estimation and uncertainty quantification under partial measurements, which is not directly applicable in the original explicit cellular automata formulation. The proposed DAE formulation, a main contribution of this manuscript, is described in the subsequent sections.

\subsubsection{Semi-explicit DAE formulation}\label{subsubsec:dae_form}
The coupled differential algebraic watershed model can be written in a compact descriptor form as follows
\begin{equation}\label{eq:dae_system}
\m E\, \dot{\m x} = \tilde{\m \psi}\bigl(\m x, \m u\bigr) = \bmat{\m f(\m x_d, \m x_a, \m u)\\ \m g(\m x_d, \m x_a)},
\end{equation}
where matrix $\m E \in \Rn{n_x \times n_x}$ represents a singular mass matrix. The mapping $\tilde{\m \psi}(\m x, \m u):= \tilde{\m \psi}\bigl(\m x(t), \m p, \m u(t)\bigr):\R\times\Rn{n_x}\times\Rn{n_p}\times\Rn{n_u}\rightarrow\Rn{n_x}$ provides a compact representation of~\eqref{eq:semi_NDAE_rep}, where the mapping depends on constant domain parameters given by vector $\m p \in \Rn{n_p}$. Note that, the mass matrix $\m E$ has a block diagonal structure with identity blocks for the differential states and zero blocks for the algebraic states. For the numerical discretization introduced in Section~\ref{subsubsec:bdf_discretization}, we use the equivalent decomposition $\tilde{\m \psi}\bigl(\m x, \m u\bigr):=\m A\m x+{\m \psi}\bigl(\m x, \m u\bigr)$, where $\m A :=\m A(t)\in\Rn{n_x\times n_x}$ collects linear time-varying terms and ${\m \psi}(\cdot)$ collects the remaining nonlinear terms.

Having presented the descriptor DAE form, we define the directional discharge rates that construct inflow and outflow coupling terms between cells $(i,j)\in\mc{K}$. For each central cell $(i,j)$, the von Neumann neighborhood $\mc{N}_{i,j}$ contains the four connected neighboring cells along $d\in\mc{D}_{\mr{dir}}=\{\mr{L},\mr{R},\mr{U},\mr{D}\}$. Let $\m Q_d(t)\in\Rn{K}$ concatenate directional discharges for each $d\in\mc{D}_{\mr{dir}}$, defined as $\m Q_d(t):=\{Q_d^{i,j}(t)\}_{(i,j)\in\mc{K}}$. Let $\{\m B_d\in\Rn{K\times K}\}_{d\in\mc{D}_{\mr{dir}}}$ denote routing matrices, considered sparse, that map directional flow rates from neighboring cell $(r,s)\in\mc{N}_{i,j}$ to the inflow of cell $(i,j)$ along direction $d$, where $[\m B_d]_{(i,j),(r,s)}=1$ when such directional connection exists, and $[\m B_d]_{(i,j),(r,s)}=0$ otherwise. Then
\begin{equation}\label{eq:inflow}
\m Q^{\mr{in}}(t):=\sum_{d\in\mc{D}_{\mr{dir}}}\m B_d\,\m Q_d(t),\qquad
\m Q^{\mr{out}}(t):=\sum_{d\in\mc{D}_{\mr{dir}}}\m Q_d(t),
\end{equation}
where vectors $\m Q^{\mr{in}}(t),\m Q^{\mr{out}}(t)\in\Rn{K}$ represent total inflow and outflow discharge rates over all the cells $(i,j)\in\mc{K}$ and considering each of the directions $d\in\mc{D}_{\mr{dir}}$.

With that in mind, the differential subsystem, $\m f(\m x_d, \m x_a, \m u)$, representing water depth and cumulative infiltration in continuous time can be written as
\begin{subequations}\label{eq:diff_states}
\begin{align}
\dot{\m h}(t) &= \m R(t)-\m f(t)+\m A_{c}^{-1}\bigl(\m Q^{\mr{in}}(t)-\m Q^{\mr{out}}(t)\bigr) \; \in \Rn{K}, \label{eq:depth_update}\\
\dot{\m F}(t) &= \m f(t) \; \in \Rn{K}, \label{eq:infil_update}
\end{align}
\end{subequations}
where Equation~\eqref{eq:depth_update} is the vector depicting the local conservation equation introduced in~\eqref{eq:mass_balance}, coupling rainfall input, infiltration, and net directional flow; Equation~\eqref{eq:infil_update} is the vector depicting cumulative infiltration dynamics introduced through~\eqref{eq:infiltration_rate}. Matrix $\m A_c:=\mr{diag}\bigl(\{A_{i,j}\}_{(i,j)\in\mc{K}}\bigr)\in\Rn{K\times K}$ is a diagonal matrix representing the cell areas, $\m R(t):=\{R^{i,j}(t)\}_{(i,j)\in\mc{K}}\in\Rn{K}$ is rainfall intensity, $\m Q^{\mr{in}}(t),\m Q^{\mr{out}}(t)\in\Rn{K}$ follow~\eqref{eq:inflow}, and $\m f(t):=\{f^{i,j}(t)\}_{(i,j)\in\mc{K}}\in\Rn{K}$ is the infiltration rate given by~\eqref{eq:infiltration_rate}. The above differential equations are thus defined for all cells $(i,j)\in\mc{K}$ that belong to the domain $\mc{D}$.

The algebraic subsystem, $\m g(\m x_d, \m x_a)$, enforcing directional Manning discharge constraints can be written as follows. Using the directional Manning cellular automata map induced by~\eqref{eq:manning}, the directional discharge flow rates, over all cells $(i,j)\in\mc{K}$, can be written as
\begin{equation}\label{eq:algebraic_constraint}
\m 0 = \m Q_d(t) - \m \phi_d\bigl(\m h(t),\m z,\m n\bigr),\qquad d\in\mc{D}_{\mr{dir}},
\end{equation}
where $\m\phi_d:\Rn{K}\times\Rn{K}\times\Rn{K}\rightarrow\Rn{K}$ denotes the directional Manning operator over all cells, with $\m h(t):=\{h^{i,j}(t)\}_{(i,j)\in\mc{K}}\in\Rn{K}$, $\m z:=\{z^{i,j}\}_{(i,j)\in\mc{K}}\in\Rn{K}$, and $\m n:=\{n^{i,j}\}_{(i,j)\in\mc{K}}\in\Rn{K}$. As introduced in Section~\ref{subsubsec:manning}, directional discharges are restricted to admissible downslope directions and are nonnegative. The slope based partitioning is detailed in~\ref{apndx:manning_map}.

Accordingly, the differential and algebraic state vectors, and the input vector, can be written as
\begin{subequations}\label{eq:state_partition}
\begin{align}
\m x_d &:= \bigl[\m h^\top(t),\,\m F^\top(t)\bigr]^\top \in \Rn{n_d}, \label{eq:xd_def}\\
\m x_a &:= \bigl[\m Q_{\mr{L}}^\top(t),\,\m Q_{\mr{R}}^\top(t),\,\m Q_{\mr{U}}^\top(t),\,\m Q_{\mr{D}}^\top(t)\bigr]^\top \in \Rn{n_a}, \label{eq:xa_def}\\
\m u &:= \m R(t):=\{R^{i,j}(t)\}_{(i,j)\in\mc{K}} \in \Rn{n_u}. \label{eq:u_def}
\end{align}
\end{subequations}
where $n_d=2K$, $n_a=4K$, and $n_x=n_d+n_a$, with $\m x:=\bigl[\m x_d^\top,\,\m x_a^\top\bigr]^\top \in\Rn{n_x}$. Here, the input vector $\m u$ is the rainfall forcing over all cells in the domain, and $n_u=K$. Vector $\m x_d$ depicts storage states over all cells through surface water depth and cumulative infiltration, while vector $\m x_a$ depicts routing discharge states through directional Manning flows. Under such a representation, the differential and algebraic states are solved simultaneously at each time step, with the differential states evolving through the mass and infiltration dynamics and the algebraic states satisfying the hydraulic discharge constraints. That being said, the algebraic states are implicitly defined by the differential states through the nonlinear Manning relation in~\eqref{eq:algebraic_constraint}, and the solution must satisfy both the mass and infiltration dynamics and the flow constraints simultaneously at each time step over all the cells. The subsequent sections introduce the numerical discretization of the proposed DAE formulation, while presenting its structural properties under physical assumptions.

\subsubsection{Solvability and numerical consideration}\label{subsubsec:index_regularity}
Before introducing the numerical discretization of the proposed DAE formulation, we show that the system is of index one and regular; this ensures local solvability of the algebraic constraints and well-posedness of the dynamics under consistent initial conditions. For the semi-explicit system in~\eqref{eq:dae_system}, the differentiation index is the minimum number of differentiation of the algebraic constraints required to obtain an ODE form~\cite{Volker2005}. In addition, regularity ensures locally well-posed dynamics for consistent initial conditions~\cite{Linh2009}. The definitions of differentiation index and regularity are given as follows.

\begin{mydef}[Differentiation index~\cite{Volker2005}]\label{def:dae_diff_index}
For a DAE system, the differentiation index $\nu$ is the smallest nonnegative integer such that, after differentiating algebraic constraints $\nu$ times, the system can be written in ODE form.
\end{mydef}

\begin{mydef}[Regularity~\cite{Volker2005}]\label{def:dae_regularity}
Consider the linearized descriptor form $\m E\dot{\m x}=\m A\m x+\m B\m u$ around an operating point. The DAE is regular if and only if the matrix pencil $(s\m E-\m A)$ is regular, i.e., $\det(s\m E-\m A)\not\equiv 0$ for $s\in\mathbb{C}$.
\end{mydef}
Definitions~\ref{def:dae_diff_index}-\ref{def:dae_regularity} are standard in DAE theory and are relevant to many physically derived models, including power system dynamics and water distribution hydraulic models. Based on the above definitions, we consider the following condition to hold true for the proposed watershed DAE model.

\begin{asmp}[Local solvability and regularity]\label{asmp:dae_regular}
For every admissible state pair $(\m x_d,\m x_a)$, with $\m h\ge\m 0$ and $\m F\ge\m 0$, the following conditions hold: \textit{(i)} topography $\m z$, cell areas $\m A_c$, and routing operators $\{\m B_d\}_{d\in\mc{D}_{\mr{dir}}}$ are fixed over each time step; \textit{(ii)} local parameters satisfy $n^{i,j}>0$, $K_s^{i,j}>0$, and $\psi_f^{i,j}>0$ for all $(i,j)\in\mc{K}$; \textit{(iii)} the differential operators are locally Lipschitz on the admissible set; and \textit{(iv)} the algebraic Jacobian with respect to algebraic states is nonsingular, that is,
\begin{equation*}
\det\left(\frac{\partial \m g}{\partial \m x_a}(\m x_d,\m x_a)\right)\neq 0.
\end{equation*}
\end{asmp}
Assumption~\ref{asmp:dae_regular} is imposed for both physical and numerical reasons; it ensures local solvability of the algebraic constraints and supports regularity of the corresponding linearized descriptor form for consistent initial conditions. From a practical viewpoint, terrain, grid geometry, and the routing operators do not vary within a simulation step, while positivity of $n$, $K_s$, and $\psi_f$ reflects physically admissible system parameters. The admissible state conditions $\m h\ge\m 0$ and $\m F\ge\m 0$ enforce nonnegative depth and cumulative infiltration. The local Lipschitz condition excludes discontinuous local behavior and is required for solvability of the implicit Newton step and validity of the realtime uncertainty propagation introduced in Section~\ref{subsec:pse}. To ensure continuous Lipschitz behavior, we introduce regularization in the infiltration relation~\eqref{eq:infiltration_rate} and the Manning relation~\eqref{eq:manning}, which prevents discontinuities near zero depth and zero cumulative infiltration, as detailed in~\ref{apndx:regularization_terms}. The following proposition establishes the index-1 structure and local algebraic solvability of the proposed DAE.

\begin{myprs}[Semi-explicit index-1 structure]\label{thm:index1_local}
Under Assumption~\ref{asmp:dae_regular}, the proposed DAE system~\eqref{eq:diff_states}-\eqref{eq:algebraic_constraint} is index-1. In particular, at each time step, the algebraic states are uniquely determined by the differential states through the algebraic constraint in~\eqref{eq:algebraic_constraint}.
\end{myprs}
Proposition~\ref{thm:index1_local} establishes a semi-explicit structure for the coupled watershed DAE under Assumption~\ref{asmp:dae_regular}. Consequently, the algebraic discharge states are uniquely recovered from the algebraic constraints at each time step, ensuring consistency between the mass and infiltration dynamics and the flow constraints. A detailed proof of the index-1 result is provided in~\ref{apndx:index_dae}. With that in mind, the subsequent section introduces a discrete-time numerical integration scheme used to advance the coupled differential and algebraic states in time while enforcing the DAE constraints at each step.

\subsubsection{Solving the dynamics in discrete-time}\label{subsubsec:bdf_discretization}
To solve the proposed DAE in discrete-time, we consider an implicit discretization based on the backward differentiation formula (BDF) method, also referred to as Gear's method~\cite{Gear1971a}. The BDF method discretizes the time derivative using current and previous state values while enforcing the algebraic constraints at the current step. This implicit approach allows for larger time steps compared to explicit methods while maintaining stability, in particular for stiff systems including the proposed coupled watershed model. The stiffness of this system arises from the presence of fast overland flow dynamics coupled with slower infiltration dynamics, as well as from the nonlinear algebraic constraints that enforce the nonlinear Manning flow relations. That being said, we consider the BDF method of order $1$, which corresponds to backward Euler discretization.

Let $k\in\mbb{N}$ denote the discrete-time index, and let the time step be denoted as $\Delta t_k>0$. Then, time at step $k$ is given by $t_k$, and $t_{k+1}$ is the time at the next step, such that $t_{k+1}:=t_k+\Delta t_k$. The first order BDF discretization of DAE~\eqref{eq:dae_system} can then be written in discrete-time DAE form as follows
\begin{equation}\label{eq:bdf_descriptor}
\m E\frac{\m x_{k+1}-\m x_k}{\Delta t_k}
=
\m A_{k+1}\m x_{k+1}
+\m \psi\bigl(\m x_{k+1},\m u_{k+1}\bigr).
\end{equation}
Equation~\eqref{eq:bdf_descriptor} is implicit in $\m x_{k+1}$; therefore, at each time step we solve a nonlinear algebraic system written in residual form as
\begin{equation}\label{eq:bdf_descriptor_residual}
\m r_{k+1}
:=
\m E(\m x_{k+1}-\m x_k)
-\Delta t_k\left(\m A_{k+1}\m x_{k+1}
+\m \psi_{k+1}\right)=\m 0,
\end{equation}
where $\m r_{k+1}:= \m r(\m x_{k+1}):\Rn{n_x}\rightarrow\Rn{n_x}$ defines the nonlinear residual mapping function at time step $k+1$ over all cells $(i,j)\in\mc{K}$ in the domain. To ensure solution convergence at each time step, the Jacobian of the nonlinear residual dynamics in~\eqref{eq:bdf_descriptor_residual} is evaluated for each Newton-Raphson iteration. At iteration $(i)$, the increment $\Delta \m x_{k+1}^{(i)}$ is computed from the Jacobian system and used to update the state as $\m x_{k+1}^{(i+1)}=\m x_{k+1}^{(i)}+\Delta \m x_{k+1}^{(i)}$. This is repeated until the convergence criterion is satisfied, then time step $k+1$ advances to the next time step until the full simulation horizon $N_T\in\mbb{N}$ is completed. The iteration increment $\Delta \m x_{k+1}^{(i)}$ can be computed by solving the following equality at each iteration 
\begin{equation}\label{eq:newton_iteration}
\m J_{k+1}^{(i)}\Delta \m x_{k+1}^{(i)}
=\m r\left(\m x_{k+1}^{(i)}\right),
\end{equation}
where the Jacobian $\m J_{k+1}^{(i)}\in\Rn{n_x\times n_x}$ of the residual mapping at iteration $(i)$ is defined as follows
\begin{equation}\label{eq:newton_jac}
\m J_{k+1}^{(i)}
:=
\left.\frac{\partial \m r}{\partial \m x}\right|_{\m x_{k+1}^{(i)}}
=
\m E
-\Delta t_k\left(
\m A_{k+1}
+\left.\frac{\partial \m \psi_{k+1}}{\partial \m x}\right|_{\m x_{k+1}^{(i)}}
\right).
\end{equation}
The iteration is repeated until the increment norm satisfies $\|\Delta \m x_{k+1}^{(i)}\|_2\le\varepsilon_{\mr{NR}}$, where $\varepsilon_{\mr{NR}}>0$ is the Newton tolerance. We note that the Jacobian $\m J_{k+1}^{(i)}$ is sparse due to local von Neumann coupling in the routing terms, which enables efficient solution of the Newton increment, although the number of states can be large for fine spatial discretizations. In the case studies, we compare the proposed approach with \textit{(i)} the explicit CA-routing formulation of HydroPol2D~\cite{Gomes2023}, and \textit{(ii)} the local inertia formulation~\cite{Gomes2024a}. The explicit CA-routing scheme advances depth under a CFL restriction and computes directional redistribution through cellular-automata weights~\cite{Gomes2023}, whereas the local inertia solver integrates a reduced shallow-water momentum balance for transient overland routing~\cite{Gomes2024a}.
\vspace{-0.2cm}
\subsubsection{Measurement under uncertainty}\label{subsubsec:output}
Based on the above state space representation of the hydrologic-hydrodynamic watershed model, we introduce the measurement equation used to represent available observations under partial watershed gauging. In practice, watershed states are only partially observed through a limited set of monitored cells, such as depth gauges at specific cells and discharge observations at flumes or weirs. The measurement model maps measured cell states to the output vector while accounting for measurement noise. The measurement equation can be written in discrete-time as follows
\begin{equation}\label{eq:output}
\m y_k = \m C\, \m x_k + \m v_k,
\end{equation}
where $\m y_k\in\Rn{n_y}$ is the vector of measured outputs, $\m C\in\Rn{n_y\times n_x}$ is the binary measurement mapping matrix, and $\m v_k\in\Rn{n_y}$ is measurement noise. Under partial watershed gauging, we typically have $n_y<n_x$, meaning only a subset of cell states is directly measured. The matrix $\m C$ is constructed by mapping available measurements to the corresponding state entries in $\m x_k$. For a cell $(i,j)\in\mc{K}$, if water depth $h^{i,j}$ is measured, the corresponding row in $\m C$ has a value of $1$ at the entry corresponding to $h^{i,j}$ in $\m x_k$ and $0$ elsewhere. If directional discharge $Q_d^{i,j}$ is measured for direction $d\in\mc{D}_{\mr{dir}}$, the corresponding row in $\m C$ has a value of $1$ at the entry corresponding to $Q_d^{i,j}$ in $\m x_k$ and $0$ elsewhere. We encode measurement availability in $\m C$ using binary indicators of whether a specific system state is measured, by defining $\gamma_h^{i,j}\in\{0,1\}$ and $\gamma_{Q_d}^{i,j}\in\{0,1\}$ for each cell $(i,j)\in\mc{K}$ and direction $d\in\mc{D}_{\mr{dir}}$, where $\gamma_h^{i,j}=1$ indicates that depth $h^{i,j}$ is measured and $\gamma_{Q_d}^{i,j}=1$ indicates that directional discharge $Q_d^{i,j}$ is measured. This allows representing realistic measurement scenarios, such as depth gauges at specific cells and discharge observations at flumes or weirs, while maintaining a single state space output model for control and estimation under partial measurements.

\subsection{A watershed's probabilistic response to uncertainty}\label{subsec:pse}
\begin{table}[t]
\centering
\caption{Sources of uncertainty and their distributions and bounds for the considered watershed model. \comment{We note that the distributions in this table are considered in~\eqref{eq:normal_propagation} to convert $\mr{K}_{\m x\m x}[k]$ into confidence intervals; they do not enter the covariance propagation~\eqref{eq:Kxx}, which uses only the means and variances of the uncertain inputs and parameters.}}
\label{tab:pse_uncertainty}
\centering
\begin{tabular}{ll}
\hline
Symbol & Distribution / bounds \\
\hline
$R$ & Log-normal, $R\in\R_{\geq 0}$ \\
$n$ & Log-normal, $n\in\R_{\geq 0}$ \\
$K_s$ & Log-normal, $K_s\in\R_{\geq 0}$ \\
$\psi_f$ & Log-normal, $\psi_f\in\R_{\geq 0}$ \\
$\eta\;(\equiv\Delta\theta)$ & Truncated Gaussian, $\eta\in[0,1]$ \\
\toprule \bottomrule
\end{tabular}
\end{table}

This section presents the probabilistic framework used to quantify the hydrologic and hydrodynamic response to uncertainty under partial measurements for the proposed DAE in Section~\ref{subsec:dae}. The objective is to characterize how uncertainty in rainfall forcing and watershed parameters propagates to realtime watershed conditions: water depth $\m h$, directional discharges $\m Q_d$, and cumulative infiltration $\m F$. Let the available observations be represented by~\eqref{eq:output}, then, given an operating trajectory $\{\m x^{0}(k)\}_{k=1}^{N_T}$ obtained by solving the proposed DAE formulation, the proposed framework computes one step covariance matrices $\mr{K}_{\m x\m x}(k)\in\Rn{n_x\times n_x}$ and corresponding confidence intervals for both measured and unmeasured watershed state variables. This entails solving the following optimization or feasibility problem.
\vspace{-0.2cm}
\begin{problem}[Probabilistic state estimation (PSE) covariance bounds]\label{prob:pse_cov_form}
The PSE covariance formulation can be written as the following problem.
\begin{subequations}\label{eq:pse_cov_form}
\begin{align}
&\mr{find}\hspace{7.2em}\mr{K}_{\m x\m x}(k); \label{eq:pse_cov_form_obj}\\
&\st \hspace{2.55em}\mr{Cov}\!\left(\mr{DAE}~\eqref{eq:bdf_descriptor},\,\mr{Measurements}~\eqref{eq:output}\right). \label{eq:pse_cov_form_con}
\end{align}
\end{subequations}
\end{problem}
Problem~\ref{prob:pse_cov_form} finds a feasible covariance matrix associated with the nonlinear DAE and measurement model. Solving the above feasibility problem is nonconvex; its solution is generally intractable. With that in mind, based on the methods presented in~\cite{Wang2022}, we reformulate the above problem by deriving a closed-form expression for $\mr{K}_{\m x\m x}$ that provides realtime confidence intervals over the nominal trajectory $\{\m x^{0}_k\}_{k=1}^{N_T}$ while accounting for parameter and measurement uncertainties. The following subsections present the considered uncertainty sources for the watershed model and then construct, through linearization around $\m x^0_k$, the closed-form covariance propagation in state space form.

\subsubsection{Considered sources of uncertainty}\label{subsubsec:uncertainty_sources}
We consider uncertainty in rainfall forcing and hydraulic parameters. The infiltration moisture-deficit parameter is denoted by $\eta$, where $\eta \equiv \Delta\theta$ in the Green-Ampt notation and $\Delta\theta=\theta_s-\theta_i$. The adopted distribution families and admissible bounds for uncertain inputs and parameters are summarized in Tab.~\ref{tab:pse_uncertainty}. We consider the following assumption on the uncertain parameters.
\begin{asmp}\label{asmp:independence}
Across the domain $\mc{D}$ and for a time horizon $N_T$, the uncertain inputs $R$, $n$, $K_s$, $\psi_f$, and $\eta$ are modeled as mutually independent random variables with known means and variances.
\end{asmp}
\comment{The above assumption introduces rainfall and hydraulic parameter uncertainty contributions in the covariance mapping in~\eqref{eq:Kxx}, while neglecting correlation among uncertain inputs and parameters. This is common in hydrologic uncertainty quantification, where correlation is often difficult to characterize and computationally expensive to model in distributed settings~\cite{Vrugt2005,Zhenyao2013,Zhao2026}. While we do not model spatial correlation among the uncertain inputs and parameters in this manuscript, the proposed covariance mapping admits such correlation through nonzero off diagonal entries in the input covariance terms; however, encoding such correlation yields denser covariance terms and reduces the sparsity inherited by the covariance computation from the local connectivity of the cell graph $\mc{G}$. We choose not to encode such correlation here for the computational aspect noted above.}
\vspace{-0.1cm}
\subsubsection{Propagating uncertainty: A DAE system's Approach}\label{subsubsec:linearization}
To reformulate Problem~\ref{prob:pse_cov_form}, we derive a linearized DAE model around the nominal trajectory $\{\m x^0_k\}_{k=1}^{N_T}$ while modeling uncertain inputs and parameters as random variables in the discrete-time DAE~\eqref{eq:bdf_descriptor}. The reason for such linearization is that the nonlinear DAE dynamics are intractable for direct covariance propagation, whereas the linearized DAE allows one to derive a closed-form expression for $\mr{K}_{\m x\m x}$ by applying the law of covariance for a linear combination of random variables. The following derivation is based on the approach developed by~\cite{Wang2022} for water distribution systems, representing a class of difference algebraic equations; however, the approach herein is extended to account for spatiotemporal distributed input and parameter uncertainty by including the random vector $\m\theta_k$ in the linearized residual model. Note that, the linearization is for computing $\mr{K}_{\m x\m x}$ and thus does not affect the nominal trajectory $\{\m x^0_k\}_{k=1}^{N_T}$, which is obtained from the original nonlinear DAE in~\eqref{eq:bdf_descriptor}.

With that in mind, let $\{\m x^0_k\}_{k=1}^{N_T}$ denote the nominal discrete-time trajectory obtained by solving the nonlinear discrete-time DAE in~\eqref{eq:bdf_descriptor}. The vector collecting all the system state variables at time step $k$ is given as
\begin{equation}
\m x_k:=\left[\m h^\top(k),\m F^\top(k),\m Q_{\mr{L}}^\top(k),\m Q_{\mr{R}}^\top(k),\m Q_{\mr{U}}^\top(k),\m Q_{\mr{D}}^\top(k)\right]^\top \in \Rn{n_x},
\end{equation}
where $n_x=6K$ and $\m x=[\m x_d^\top,\m x_a^\top]^\top$. Furthermore, we define the input and parameter vector representing the uncertain inputs and parameters at time step $k$ as follows
\begin{equation}\label{eq:theta_def}
\m\theta_k:=\left[\m R^\top(k),\m n^\top,\m K_s^\top,\m\psi_f^\top,\m\eta^\top\right]^\top\in\Rn{n_\theta},
\end{equation}
where vector $\m \theta_k=\left[\m u_k^\top,\m p_{\theta}^\top\right]^\top $ concatenates the vectors of uncertain inputs $\m u_k = \m R(k) \in \Rn{K}$ and uncertain parameters $\m p_{\theta} = [\m n^\top, \m K_s^\top, \m \psi_f^\top, \m \eta^\top]^\top \in \Rn{4K}$ in the domain, and $n_\theta=5K$. 
Vector $\m p_{\theta}$ denotes the subset of uncertain hydraulic and infiltration parameters selected from $\m p$ in Section~\ref{subsubsec:dae_form}. That is, the following are considered as random variables within the domain (see Section~\ref{subsubsec:uncertainty_sources}): $\m R(k)$, $\m n$, $\m K_s$, $\m\psi_f$, and $\m\eta$, which denote the rainfall, Manning's roughness, saturated hydraulic conductivity, wetting front suction head, and moisture-deficit parameter at time step $k$ over all cells $(i,j) \in \mc{K}$ in the domain $\mc{D}$. Equation~\eqref{eq:theta_def} shows explicitly that $\m\theta_k$ collects uncertain input and uncertain parameters in one vector that is representative of the proposed DAE model defined by $\tilde{\m\psi}\bigl(\m x_k,\m p,\m u_k\bigr)$. Under such uncertain conditions, vector $\m\theta_k$ is modeled as a random vector around its nominal value $\m\theta_k^0$, where $\m\theta_k^0$ represents the expected value of the uncertain inputs and parameters.

For the nonlinear discrete-time DAE in~\eqref{eq:bdf_descriptor}, and considering the nominal trajectory $\{\m x^0_k\}_{k=1}^{N_T}$ that represents the expected or mean behavior of the watershed states under nominal input and parameter conditions $\{\m \theta^0_k\}_{k=1}^{N_T}$, the linearized DAE is obtained by applying the first order expansion around the nominal operating point of the discrete-time residual function $\m r_k \equiv \m r_k\left(\m x_k,\m x_{k-1},\m\theta_k\right)$ at time step $k$, as follows
\begin{equation}\label{eq:linearized_residual_full}
\left.\frac{\partial\m r_k}{\partial\m x_k}\right|_{\m x^0}\m x_k
+\left.\frac{\partial\m r_k}{\partial\m x_{k-1}}\right|_{\m x^0}\m x_{k-1}
+\left.\frac{\partial\m r_k}{\partial\m\theta_k}\right|_{\m x^0}\m\theta_k
=\m g_k,
\end{equation}
where from~\eqref{eq:bdf_descriptor_residual}, we have $\partial\m r_k/\partial\m x_{k-1}=-\m E$. Now, defining the Jacobian matrices for the linearized DAE associated with the residual function $\m r_k$ at step $k$ as
\begin{equation}\label{eq:Ak_Btheta_def}
\m A_k^{\m \psi}:=\left.\frac{\partial\m r_k}{\partial\m x_k}\right|_{\m x^0}\in\Rn{n_x\times n_x},\qquad
\m B_{\theta,k}:=\left.\frac{\partial\m r_k}{\partial\m\theta_k}\right|_{\m x^0}\in\Rn{n_x\times n_\theta},
\end{equation}
we obtain the following compact linear form of the nonlinear DAE for a time step $k$ as follows.
\begin{equation}\label{eq:compact_linear}
\m A_k^{\m \psi}\,\m x_k=\m E\,\m x_{k-1}+\m B_{\theta,k}\,\m\theta_k+\m b_{0,k},
\end{equation}
where $\m b_{0,k}:=\m g_k$ is a deterministic offset term that captures the residual of the linearization around the nominal trajectory.
The Jacobian $\m A_k^{\m \psi}$ can be written as
\begin{equation}\label{eq:Ak_explicit}
\m A_k^{\m \psi}=
\m E-\Delta t_k\left(\m A_k+\left.\frac{\partial\m\psi_k}{\partial\m x}\right|_{\m x^0}\right),
\end{equation}
where above expression is based on the discrete-time BDF residual in~\eqref{eq:bdf_descriptor_residual}. Note that higher order BDF methods would result in a similar form; however, the Jacobian $\m A_k^{\m \psi}$ would include additional terms from previous states $\m x_{k-2},\ldots,\m x_{k-k_g}$, where $k_g$ is the order of the BDF method. The Jacobian $\m B_{\theta,k}$ captures the change in the residual function to uncertain inputs and parameters, and its structure depends on how the uncertain inputs and parameters enter the nonlinear DAE dynamics. For example, uncertain rainfall $R$ enters through the mass balance dynamics~\eqref{eq:depth_update}, while uncertain hydraulic parameters $n$, $K_s$, $\psi_f$, and $\eta$ enter through both water depth~\eqref{eq:depth_update} and infiltration dynamics~\eqref{eq:infil_update}, as well as the algebraic flow constraints~\eqref{eq:algebraic_constraint}.

Having presented the linearized DAE model and by incorporating the measurement model in~\eqref{eq:output}, we can now derive the closed-form covariance propagation for the linearized DAE under partial measurements. To that end, let $N_k:=N_T-k+1$, where $N_k$ denotes the number of discrete-time steps in the remaining observation horizon $\{k,k+1,\ldots,N_T\}$. In what follows, $N_T$ is used for time indexing, while $N_k$ is used for horizon dimensions. The states over this horizon can be concatenated as $\m x[k]:=\{\m x_i\}_{i=k}^{N_T}\in \Rn{N_k n_x}$, representing the trajectory of states from step $k$ to the end of the simulation horizon $N_T$. Now, considering the measurements over the same horizon, $\m y[k]:=\{\m y_i\}_{i=k}^{N_T}\in \Rn{N_k n_y}$, we can rewrite the above expressions as a set of linear equations by considering the system state trajectory $\m x[k]$ and the corresponding measurement trajectory $\m y[k]$ over the horizon $N_T$, as follows
\begin{equation}\label{eq:stacked_linear}
\m A[k]\m x[k]=\m b[k],
\end{equation}
where $\m A[k]$ and $\m b[k]$ collect all linear equations over the horizon $\{k,\ldots,N_T\}$. In particular, each block row combines: \textit{(i)} the linearized DAE relation at that step; and \textit{(ii)} the measurement relation under partial observations. The matrix $\m A[k]$ representing the linear DAE~\eqref{eq:compact_linear} and measurement model~\eqref{eq:output} is written as
\begin{equation}\label{eq:Ak_structure}
\m A[k]:=
\begin{bmatrix}
\m A^{\mr{s}}_{k} & \m 0 & \cdots & \m 0\\
-\m E_s & \m A^{\mr{s}}_{k+1} & \cdots & \m 0\\
\vdots & \ddots & \ddots & \vdots\\
\m 0 & \cdots & -\m E_s & \m A^{\mr{s}}_{N_T}
\end{bmatrix}\in\Rn{N_k n_s\times N_k n_x},
\end{equation}
where $n_s:=n_x+n_y$ and the matrix $\m A^{\mr{s}}_{k}$ for time step $k$ captures the linearized DAE and measurement model and is defined as
\begin{equation}\label{eq:As_def}
\m A^{\mr{s}}_{k}:=
\begin{bmatrix}
\m A_k^{\m \psi}\\
\m C
\end{bmatrix}\in\Rn{n_s\times n_x},
\end{equation}
where $\m A_k^{\m \psi}$ is the Jacobian of the linearized residual dynamics at time step $k$, matrix $\m C$ is the measurement mapping from~\eqref{eq:output}, matrix $\m E_s:=[\m E^{\top}, \m 0^{\top}_{n_y\times n_x}]^\top\in\Rn{n_s\times n_x}$ is the mass matrix concatenated with a zero matrix. Thus, $\m A_k^{\mr{s}}$ represents the dynamics and measurement mapping at time step $k$.
Furthermore, vector $\m b[k]\in \Rn{N_k n_s}$ collects the corresponding terms associated with uncertain inputs/parameters in $\m\theta_k$ and measurement uncertainty in $\m v_k$, and is written as
\begin{equation}\label{eq:bk_stack}
\m b^{\mr{s}}_k
:=\m E_s\,\m x_{k-1}
+\m B_{\theta,k}^{\mr{s}}\,\m\theta_k
+\begin{bmatrix}\m b_{0,k}\\ \m y_k -\m v_k\end{bmatrix}
\in\Rn{n_s},
\end{equation}
where matrix $\m B_{\theta,k}^{\mr{s}} := [\m B_{\theta,k}^{\top}, \m 0^{\top}_{n_y \times n_\theta}]^\top \in \Rn{n_s \times n_\theta}$ is the forcing/parameter matrix~\eqref{eq:Ak_Btheta_def} concatenated with a zero matrix of appropriate size. Note that $\m b_k^{\mr{s}}$ is a random vector due to the presence of the uncertain input and parameter vector $\m \theta_k$ and measurement noise $\m v_k$. The uncertainty in $\m b_k^{\mr{s}}$ can be decomposed into three contributions: \textit{(i)} propagated uncertainty from the previous state $\m x_{k-1}$ through the term $\m E_s \m x_{k-1}$; \textit{(ii)} spatiotemporal forcing/parameter uncertainty through the term $\m B_{\theta,k}^{\mr{s}} \m \theta_k$; and \textit{(iii)} measurement uncertainty through the term $\m v_k$. The first two are a result of the DAE dynamics, while the third contribution is due to the measurement model. Based on such a formulation that combines the watershed dynamics under the proposed DAE representation and the measurement model under partial gauging, we can now derive the closed-form covariance computation taking into account the contributions from the above uncertainty sources. Notice that from~\eqref{eq:stacked_linear}-\eqref{eq:bk_stack}, the DAE and measurement relations in the covariance constraint~\eqref{eq:pse_cov_form_con} are explicitly represented by the linear equation $\m A[k]\m x[k] = \m b[k]$. Accordingly, the constraints in Problem~\ref{prob:pse_cov_form} can be rewritten as a feasibility problem in terms of~\eqref{eq:stacked_linear} as follows.
\vspace{-0.2cm}
\begin{problem}[PSE covariance bounds, rewritten]\label{prob:pse_cov_linear_form}
\begin{subequations}\label{eq:pse_cov_linear_form}
\begin{align}
&\mr{find}\hspace{7.2em}\mr{K}_{\m x\m x}[k], \label{eq:pse_cov_linear_obj}\\
&\st \hspace{2.55em}\mr{Cov}\!\left(\m A[k]\m x[k]=\m b[k]\right), \label{eq:pse_cov_linear_con}
\end{align}
\end{subequations}
\end{problem}\vspace{-0.1cm}
\noindent where $\mr{K}_{\m x\m x}[k]:=\mr{Cov}\!\left(\m x[k],\m x[k]\right)\in\Rn{N_k n_x\times N_k n_x}$ is the covariance matrix of the trajectory state vector $\m x[k]$, and matrix $\mr{K}_{\m b\m b}[k]:=\mr{Cov}\!\left(\m b[k],\m b[k]\right)\in\Rn{N_k n_s\times N_k n_s}$ is the covariance of the vector $\m b[k]$, both defined over the horizon $\{k,k+1,\ldots,N_T\}$. Note that Problem~\ref{prob:pse_cov_linear_form} is a feasibility problem where the covariance constraint is equivalent to the covariance relation for linear equations, which is given by
\begin{equation}\label{eq:cov_relation_linear}
\mr{Cov}\left(\m A[k]\m x[k]=\m b[k]\right) \Leftrightarrow 
\m A[k]\mr{K}_{\m x\m x}[k]\m A[k]^\top = \mr{K}_{\m b\m b}[k].
\end{equation}
The above relation is a direct consequence of the covariance relation for linear combinations of random variables, and it provides a closed-form expression for the covariance $\mr{K}_{\m x\m x}[k]$ that satisfies the covariance constraint in Problem~\ref{prob:pse_cov_linear_form}. The following section presents the detailed derivation of the covariance computation based on the above relation.
\subsubsection{Realtime covariance computation under partial measurements}\label{subsubsec:cov_prop}
By solving Problem~\ref{prob:pse_cov_linear_form}, we obtain the covariance of $\m x[k]$ over the horizon $\{k,\ldots,N_T\}$, while the variance of each state entry at time step $k$ is obtained from the diagonal entries of $\mr{K}_{\m x\m x}(k)$. The closed-form covariance solution is given by the following proposition.
\begin{myprs}\textit{(Closed-form covariance solution)}\label{thm:pse_solution}
For the DAE under uncertain inputs and parameters, the constraints in Problem~\ref{prob:pse_cov_linear_form} are represented by the stacked linear equation~\eqref{eq:stacked_linear}, with covariance form given by~\eqref{eq:cov_relation_linear}. Under Assumption~\ref{asmp:dae_regular}, the closed-form solution of $\mr{K}_{\m x\m x}[k]$ is written as
\begin{equation}\label{eq:Kxx}
\begin{split}
\mr{K}_{\m x\m x}[k] := \;& \bigl(\m A[k]^\top \m A[k]\bigr)^{-1}\,\m A[k]^\top\,\mr{K}_{\m b\m b}[k]\,\m A[k]\\
&\times \bigl(\m A[k]^\top \m A[k]\bigr)^{-1},
\end{split}
\end{equation}
where for the state trajectory $\m x[k]$ over the horizon $\{k,\ldots,N_T\}$, the covariance at each time step, denoted by $\mr{K}_{\m x\m x}(k)\in\Rn{n_x\times n_x}$, is obtained from the corresponding diagonal block of $\mr{K}_{\m x\m x}[k]$, where $\mr{K}_{\m x\m x}[k]:= \mr{diag}\!\left(\mr{K}_{\m x\m x}(k),\mr{K}_{\m x\m x}(k+1),\ldots,\mr{K}_{\m x\m x}(N_T)\right)$.
\end{myprs}
The derivation of Proposition~\ref{thm:pse_solution} is provided in~\ref{apndx:kxx_solution}. The above solution is a direct consequence of the covariance relation for linear equations in~\eqref{eq:cov_relation_linear}, where the covariance $\mr{K}_{\m x\m x}[k]$ is obtained by applying the law of covariance for linear combinations of random variables to the linearized DAE and measurement model. The covariance mapping in~\eqref{eq:Kxx} provides a closed-form solution for $\mr{K}_{\m x\m x}[k]$ that satisfies the covariance constraint in Problem~\ref{prob:pse_cov_linear_form}, and it accounts for the contributions from uncertainty in rainfall forcing and hydraulic parameters, as well as measurement uncertainty under partial gauging.

Furthermore, to compute the covariance $\mr{K}_{\m b\m b}[k]$, we concatenate the covariance of $\m b^{\mr{s}}_k$ over the horizon $\{k,\ldots,N_T\}$ according to the following block diagonal structure, where the covariance at each time step $k$ is given by $\mr{K}_{\m b^{\mr{s}}\m b^{\mr{s}}}(k)$.
\begin{equation}\label{eq:Kbb_stacked_pre}
\mr{K}_{\m b\m b}[k]= \mr{diag}\!\left(\mr{K}_{\m b^{\mr{s}}\m b^{\mr{s}}}(k),\mr{K}_{\m b^{\mr{s}}\m b^{\mr{s}}}(k+1),\ldots,\mr{K}_{\m b^{\mr{s}}\m b^{\mr{s}}}(N_T)\right).
\end{equation}
At time step $k$, using $\m b_k^{\mr{s}}$ from~\eqref{eq:bk_stack} and noting that deterministic terms do not affect covariance, the covariance $\mr{K}_{\m b^{\mr{s}}\m b^{\mr{s}}}(k)$ can be written as 
\begin{equation}\label{eq:Kbs_cov_pre}
\begin{split}
\mr{K}_{\m b^{\mr{s}}\m b^{\mr{s}}}(k)
:&={}\mr{Cov}\!\left(\m b_k^{\mr{s}},\m b_k^{\mr{s}}\right) \in\Rn{n_s\times n_s},\\
&=\m E_s\,\mr{K}_{\m x\m x}(k-1)\,\m E_s^\top
+\mr{K}_{\m b_{\theta}\m b_{\theta}}(k)
+\mr{K}_{\m v}^{\mr{s}}(k),
\end{split}
\end{equation}
where the covariance $\mr{K}_{\m b^{\mr{s}}\m b^{\mr{s}}}(k)$ combines propagated uncertainty from the previous state through $\mr{K}_{\m x\m x}(k-1)$, spatiotemporal forcing/parameter uncertainty through $\mr{K}_{\m b_{\theta}\m b_{\theta}}(k)$, and measurement uncertainty through $\mr{K}_{\m v}^{\mr{s}}(k)$. Accordingly, $\mr{K}_{\m x\m x}(k-1):=\mr{Cov}\!\left(\m x_{k-1},\m x_{k-1}\right)\in\Rn{n_x\times n_x}$ denotes the state covariance at the previous time step. The covariance contribution from uncertain inputs and parameters is given by
\begin{equation}\label{eq:Kbu_pre}
\mr{K}_{\m b_{\theta}\m b_{\theta}}(k)
:=\m B_{\theta,k}^{\mr{s}}\mr{K}_{\theta}(k){\m B_{\theta,k}^{\mr{s}}}^\top
\in\Rn{n_s\times n_s},
\end{equation}
where the covariance $\mr{K}_{\theta}(k):=\mr{Cov}\!\left(\m\theta_k,\m\theta_k\right)\in\Rn{n_\theta\times n_\theta}$ captures the uncertainty in the input and parameter vector $\m \theta_k$ at time step $k$ and is defined as follows
\begin{equation}\label{eq:Ktheta_pre}
\mr{K}_{\theta}(k)
:=\mr{diag}\!\left(
\mr{K}_R(k),\mr{K}_n(k),\mr{K}_{K_s}(k),\mr{K}_{\psi_f}(k),\mr{K}_{\eta}(k)\right),
\end{equation}
where each covariance $\mr{K}_{\{\bullet\}}(k)\in\Rn{K\times K}$ corresponds to the respective covariances of the uncertain inputs and parameters defined by vector $\m \theta_k$. Finally, the covariance from the measurement noise is defined as follows
\begin{equation}\label{eq:Kv_pre}
\mr{K}_{\m v}^{\mr{s}}(k)
:=\begin{bmatrix}
\m 0_{n_x\times n_x} & \m 0\\
\m 0 & \mr{K}_{\m v\m v}(k)
\end{bmatrix}
\in\Rn{n_s\times n_s},
\end{equation}
where $\mr{K}_{\m v\m v}(k):=\mr{Cov}\!\left(\m v_k,\m v_k\right)\in\Rn{n_y\times n_y}$ denotes the covariance of the measurement noise at time step $k$.

To that end, in order to compute Equation~\eqref{eq:Kxx}, the stacked covariance $\mr{K}_{\m b\m b}[k]$ is assembled from~\eqref{eq:Kbb_stacked_pre}, where each block covariance $\mr{K}_{\m b^{\mr{s}}\m b^{\mr{s}}}(k)$ is given by~\eqref{eq:Kbs_cov_pre} and depends on $\mr{K}_{\m x\m x}(k-1)$, $\mr{K}_{\m b_{\theta}\m b_{\theta}}(k)$ in~\eqref{eq:Kbu_pre}, with $\mr{K}_{\theta}(k)$ and $\mr{K}_{\m v}^{\mr{s}}(k)$ defined by~\eqref{eq:Ktheta_pre} and~\eqref{eq:Kv_pre}. We note that under such construction, Equation~\eqref{eq:Kxx} admits a linear least squares interpretation over the same horizon. By defining $(\m A[k])^\dagger := (\m A[k]^\top \m A[k])^{-1} \m A[k]^\top$, we obtain $\mr{K}_{\m x\m x}[k] = (\m A[k])^\dagger \mr{K}_{\m b\m b}[k] ((\m A[k])^\dagger)^\top$, which is the covariance of the estimator $\hat{\m x}[k] = (\m A[k])^\dagger \m b[k]$. This relation is similar to a linear Gaussian estimator; for Gaussian random vectors, the first two moments fully characterize the distribution, and covariance propagation is therefore sufficient to quantify the estimation uncertainty. However, in the proposed framework, covariance propagation does not require Gaussian distribution assumptions. \comment{The closed-form mapping in~\eqref{eq:Kxx} is distribution-agnostic in the covariance propagation, as it propagates the input covariance $\mr{K}_{\m b\m b}[k]$ to the state covariance $\mr{K}_{\m x\m x}[k]$ using only the means and variances of rainfall, parameters, and measurement noise, without assuming a specific probability distribution for any of these inputs. A distributional assumption is introduced only at the confidence interval computation in~\eqref{eq:normal_propagation}, where the variance $[\mr{K}_{\m x\m x}(k)]_{jj}$ for each state variable $x_j(k)$ is mapped to a confidence interval using the adopted distribution model $\mc{D}_j$ for each variable in Tab.~\ref{tab:pse_uncertainty}. The two operations are independent, and $\mr{K}_{\m x\m x}[k]$ does not depend on the choice between Gaussian, log-normal, and truncated Gaussian families; the adopted distribution model can be replaced without recomputing $\mr{K}_{\m x\m x}[k]$.} That being said, unlike conventional sensitivity based first and second order moment mappings, the proposed DAE formulation considers $\m A[k]$ and propagates uncertainty across all states (measured and unmeasured) while accounting for all sources and types of uncertainty in the DAE and measurement model. 

\begin{algorithm}[t]
\caption{State Space Uncertainty Propagation Algorithm for Coupled Overland Flow and Infiltration Dynamics}\label{algorithm1}
\DontPrintSemicolon
\textbf{input:} nominal trajectory $\{\m x^0_k\}_{k=1}^{N_T}$\;
\textbf{input:} rainfall and parameter means\;
\textbf{input:} $\mr{Var}(\m R(k)),\mr{Var}(\m n),\mr{Var}(\m K_s),\mr{Var}(\m \psi_f),\mr{Var}(\m \eta),\mr{Var}(\m v_k)$\;
\textbf{input:} time horizon $N_T$\;
\textbf{initialize:} $\mr{K}_{\m x\m x}(0)\leftarrow \m 0$ // \textit{initial covariance at the initial step}\;
\For{$k = 1, 2, \ldots, N_T$}{
// \textit{linearization and covariance propagation at time index $k$}\;
\textbf{linearize:} DAE residual at $\m x^0_k$ and build $\m A[k]$\;
\textbf{assemble:} $\mr{K}_{\m b^{\mr{s}}\m b^{\mr{s}}}(k)$ and $\mr{K}_{\m b\m b}[k]$ via~\eqref{eq:Kbs_cov_pre} and~\eqref{eq:Kbb_stacked_pre} // \textit{build uncertainty contributions from $\m\theta$ and $\m v$}\;
\textbf{compute:} $\mr{K}_{\m x\m x}[k]$ via Proposition~\ref{thm:pse_solution},~\eqref{eq:Kxx} // \textit{solve covariance mapping}\;
\textbf{extract:} $\mr{K}_{\m x\m x}(k)$ as the leading $n_x\times n_x$ block of $\mr{K}_{\m x\m x}[k]$ // \textit{state covariance at time step $k$}\;
\textbf{extract:} state variances from $\mr{diag}\bigl(\mr{K}_{\m x\m x}(k)\bigr)$ // \textit{variance of each state variable}\;
}
\textbf{compute:} confidence intervals from configured distributions // \textit{for measured and unmeasured states}\;
\textbf{output:} $\{\mr{K}_{\m x\m x}(k)\}_{k=1}^{N_T}$ and intervals for all six states\;
\end{algorithm}

Furthermore, at each time step $k$, we stack $n_x$ linearized DAE equations with $n_y$ measurement equations, yielding $n_s = n_x + n_y$ equations for $n_x$ unknown states. Thus, under measurement, i.e., $n_y > 0$, the linear system $\m A[k] \m x[k] = \m b[k]$ is overdetermined, meaning that there are more equations than unknowns. This allows the covariance solution in~\eqref{eq:Kxx} to leverage the measurement information to reduce estimation uncertainty over the horizon $\{k,\ldots,N_T\}$, which is analogous to the covariance reduction effect of a Kalman update. However, for Kalman filters, this correction is performed recursively through a prediction correction update using the mismatch between measured and model predicted outputs~\cite{Kalman1960,Bartos2021a}. With that distinction, and based on Proposition~\ref{thm:pse_solution}, Algorithm~\ref{algorithm1} computes covariances for the state trajectory $\m x[k]$ over the horizon $\{k,\ldots,N_T\}$ by applying the closed-form mapping in~\eqref{eq:Kxx} to the stacked linearized DAE and measurement model, which is represented by $\m A[k]$ and $\mr{K}_{\m b\m b}[k]$.

The inputs for Algorithm~\ref{algorithm1} are the nominal trajectory $\{\m x_k^0\}_{k=1}^{N_T}$, the mean values of rainfall and hydraulic infiltration parameters, the selected uncertainty distributions as presented in Tab.~\ref{tab:pse_uncertainty}, the variance terms $\mr{Var}(\m R(k))$, $\mr{Var}(\m n)$, $\mr{Var}(\m K_s)$, $\mr{Var}(\m \psi_f)$, $\mr{Var}(\m \eta)$, and $\mr{Var}(\m v_k)$, and the simulation horizon $N_T$. \comment{Fig.~\ref{fig:workflow_pse} summarizes Algorithm~\ref{algorithm1}, depicting the inputs, the covariance terms $\mr{K}_{\m b^{\mr{s}}\m b^{\mr{s}}}(k)$ and $\mr{K}_{\m b\m b}[k]$ assembled at each time step, and the resulting state covariance $\mr{K}_{\m x\m x}[k]$ from which the confidence intervals are computed.} At each time step $k$, the algorithm linearizes the DAE around $\m x_k^0$, constructs the stacked linear system matrices, assembles the uncertainty covariance terms $\mr{K}_{\m b^{\mr{s}}\m b^{\mr{s}}}(k)$ and $\mr{K}_{\m b\m b}[k]$ from forcing/parameter and measurement uncertainty, and then computes $\mr{K}_{\m x\m x}[k]$ based on Equation~\eqref{eq:Kxx}. The covariance $\mr{K}_{\m x\m x}(k)$ at each time step is then extracted from the horizon covariance, and the corresponding state variances are obtained from its diagonal entries. The output of Algorithm~\ref{algorithm1} is the sequence $\{\mr{K}_{\m x\m x}(k)\}_{k=1}^{N_T}$ together with confidence intervals for water depth, directional discharges, and cumulative infiltration for both measured and unmeasured states. Thus, the algorithm provides a unified realtime state space covariance computation under partial watershed measurements.

The covariance computation in Section~\ref{subsubsec:cov_prop} provides $\mr{K}_{\m x\m x}(k)$ at each time step through the closed-form mapping in~\eqref{eq:Kxx}. Confidence intervals for measured and unmeasured states are then computed around the nominal trajectory $\{\m x_k^0\}_{k=1}^{N_T}$ using the adopted distribution model (Gaussian, log-normal, or truncated Gaussian). For each state variable $x_j(k)$, $j \in \{1,\ldots,n_x\}$, we can compute the confidence interval as follows
\begin{equation}\label{eq:normal_propagation}
x_j(k) \sim \mc{D}_j\!\left(x_j^0(k),\,\left[\mr{K}_{\m x\m x}(k)\right]_{jj}\right),
\end{equation}
where $\left[\mr{K}_{\m x\m x}(k)\right]_{jj}$ is the variance of the $j$th state variable. For the Gaussian distribution, $\mc{D}_j = \mc{N}$, and Equation~\eqref{eq:normal_propagation} reduces to $x_j(k) \sim \mc{N}\!\left(x_j^0(k), \left[\mr{K}_{\m x\m x}(k)\right]_{jj}\right)$. Thus, for each state variable in the system, we can compute a confidence interval at each time step $k$ based on the variance obtained from the covariance mapping in~\eqref{eq:Kxx} and the nominal trajectory. This allows quantifying uncertainty in both measured and unmeasured states under the proposed DAE framework, while accounting for all sources of uncertainty in the watershed dynamics model.

\begin{figure}[t]
\centering
\providecommand{\figeqref}[1]{\eqref{#1}}
\providecommand{\figsecref}[1]{\ref{#1}}
\resizebox{\columnwidth}{!}{%
\begin{tikzpicture}[
font=\small,
>={Stealth[length=2.2mm]},
stepbox/.style={
rectangle, rounded corners=2pt, draw=black!70, line width=0.5pt,
minimum height=18mm, minimum width=30mm, align=center,
inner sep=2.5pt, fill=#1, text width=28mm
},
outbox/.style={
rectangle, rounded corners=2pt, line width=0.5pt,
fill=white, align=center, inner sep=2.5pt,
font=\footnotesize\itshape, text width=28mm
},
flowarrow/.style={->, line width=0.7pt, draw=black!75},
feedbackarrow/.style={->, dashed, draw=black!55, line width=0.5pt},
]
\definecolor{Sub1purple}{RGB}{200,200,255}
\definecolor{Sub2yellow}{RGB}{251,234,206}
\definecolor{Sub3green}{RGB}{215,238,210}
\colorlet{phys}{Sub3green}
\colorlet{nominal}{Sub1purple}
\colorlet{linbox}{Sub1purple}
\colorlet{asmbox}{Sub2yellow}
\colorlet{covbox}{Sub2yellow}
\colorlet{ciband}{Sub3green}
\node[stepbox=phys] (s1) {%
\textbf{1.~DAE}~\figeqref{eq:semi_NDAE_rep}\\[1pt]
$\m E\,\dot{\m x}{=}\m f(\m x,\m\theta,R)$\\
$\m 0{=}\m g(\m x,\m\theta)$\\[1pt]
{\footnotesize $\m\mu_{\m\theta}$, $\mr{K}_{\theta}$~\figeqref{eq:Ktheta_pre}}
};
\node[stepbox=nominal, right=4mm of s1] (s2) {%
\textbf{2.~Nominal}\\[1pt]
BDF--Newton~\figeqref{eq:newton_iteration}\\[1pt]
$\{\m x_k^{0}\}_{k=1}^{N_T}$
};
\node[stepbox=linbox, right=4mm of s2] (s3) {%
\textbf{3.~Linearize}\\[1pt]
$\m A_k^{\m\psi},\;\m B_{\theta,k}$~\figeqref{eq:Ak_Btheta_def}\\[-1pt]
$\m A[k]\m x[k]{=}\m b[k]$\\[-1pt]
{\footnotesize\figeqref{eq:Ak_structure}}
};
\node[stepbox=asmbox, below=10mm of s3] (s4) {%
\textbf{4.~Assemble}\\[1pt]
$\mr{K}_{\m b\m b}[k]$\\[1pt]
{\footnotesize $\mr{K}_{\m x\m x}(k{-}1)$, $\mr{K}_{\theta}(k)$, $\mr{K}_{\m v\m v}(k)$}\\[-1pt]
{\footnotesize\figeqref{eq:Kbs_cov_pre},~\figeqref{eq:Kbb_stacked_pre}}
};
\node[stepbox=covbox, left=4mm of s4] (s5) {%
\textbf{5.~Covariance}\\[1pt]
$\mr{K}_{\m x\m x}[k]$~\figeqref{eq:Kxx}\\[1pt]
{\footnotesize Prop.~\figsecref{thm:pse_solution}}\\[-1pt]
{\footnotesize Extract $\mr{K}_{\m x\m x}(k)$}
};
\node[stepbox=ciband, left=4mm of s5] (s6) {%
\textbf{6.~Intervals}\\[1pt]
$\m x_k^{0}\pm z_{\alpha}\sqrt{\mathrm{diag}(\mr{K}_{\m x\m x}(k))}$\\[1pt]
{\footnotesize\figeqref{eq:normal_propagation}}\\[-1pt]
{\footnotesize Gauged \& ungauged}
};
\node[outbox, above=0mm of s6] (out) {%
Realtime PSE at all cells in $\mc{D}$
};
\draw[flowarrow] (s1) -- (s2);
\draw[flowarrow] (s2) -- (s3);
\draw[flowarrow] (s3) -- node[right, font=\footnotesize\itshape, text=black!65]
{for $k{=}1,\ldots,N_T$} (s4);
\draw[flowarrow] (s4) -- (s5);
\draw[flowarrow] (s5) -- (s6);
\draw[feedbackarrow] (s5.south) -- ++(0,-3.5mm) -| node[below, pos=0.25, font=\footnotesize]
{$\mr{K}_{\m x\m x}(k{-}1)$} (s4.south);
\end{tikzpicture}}
\caption{\comment{The proposed covariance propagation framework implemented by Algorithm~\ref{algorithm1}. The nonlinear DAE in~\eqref{eq:semi_NDAE_rep} generates the nominal trajectory $\{\m x_k^{0}\}_{k=1}^{N_T}$ through the BDF Newton iteration in~\eqref{eq:newton_iteration} and is linearized at $\m x_k^{0}$ to obtain the Jacobians in~\eqref{eq:Ak_Btheta_def}. The linearized DAE is stacked over the horizon $N_T$ together with the partial gauge measurements to form $\m A[k]\,\m x[k]=\m b[k]$ in~\eqref{eq:Ak_structure}, the uncertainty covariance terms $\mr{K}_{\m b^{\mr{s}}\m b^{\mr{s}}}(k)$ and $\mr{K}_{\m b\m b}[k]$ are assembled via~\eqref{eq:Kbs_cov_pre} and~\eqref{eq:Kbb_stacked_pre}, the closed-form covariance mapping in~\eqref{eq:Kxx} produces $\mr{K}_{\m x\m x}[k]$, and confidence intervals over $\m x_k^{0}$ are computed via~\eqref{eq:normal_propagation}.}}
\label{fig:workflow_pse}
\end{figure}

\section{Case study areas and data}\label{sec:study_areas}
In this manuscript, we consider two case studies to evaluate the proposed framework. The first case study is the synthetic V-Tilted catchment, which has been widely adopted as a benchmark for surface runoff models~\cite{Sulis2010,Kim2012a,Cozzolino2019,Gomes2023,Meles2024}. The second case study is the Walnut Gulch experimental watershed (WGEW), located in southeastern Arizona, USA~\cite{Skirvin2008,Nichols2008,Polyakov2018,Becker2018,Kautz2019,Meles2021,Meles2024}.

\subsection{V-Tilted Catchment}\label{subsec:vtilted_case_study} The synthetic V-Tilted catchment is depicted in Fig.~\ref{fig:vtitled}. This catchment consists of two rectangular hillslope planes ($800\;\mr{m} \times 1000\;\mr{m}$ each) coupled with a vegetated channel ($20\;\mr{m} \times 1000\;\mr{m}$) along the central axis. The total area of the catchment is $1620 \;\mr{m} \times 1000 \;\mr{m}=1.62\;\mr{km}^2$. The slope in the transverse direction ($x\mr{-}x$ direction) is $5\%$, while along the longitudinal direction ($y\mr{-}y$ direction) the slope is $2\%$. Two types of ground covers are considered with Manning roughness given as $n = 0.015\;\mr{s \cdot m^{-1/3}}$ for the hillslopes and $n = 0.15\;\mr{s \cdot m^{-1/3}}$ for the channel.

\begin{figure}[t]
\centering
\includegraphics[keepaspectratio=true,width=\columnwidth]{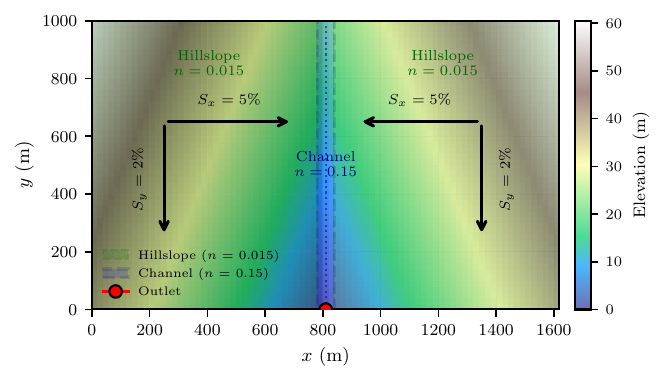}
\vspace{-0.7cm}
\caption{Numerical case study catchment (V-Tilted): synthetic topography with smooth hillslopes and a rougher central channel. The hillslopes and channel are assigned Manning roughness values of $n=0.015$ and $n=0.15$. The cross-slopes are $S_x=5\%$ and $S_y=2\%$. The outlet boundary condition is prescribed as normal depth with slope $0.02$. The grid resolution is $20~\mr{m}$.}\label{fig:vtitled}
\end{figure}

\begin{figure*}[t]
\centering
\includegraphics[keepaspectratio=true,width=1\textwidth]{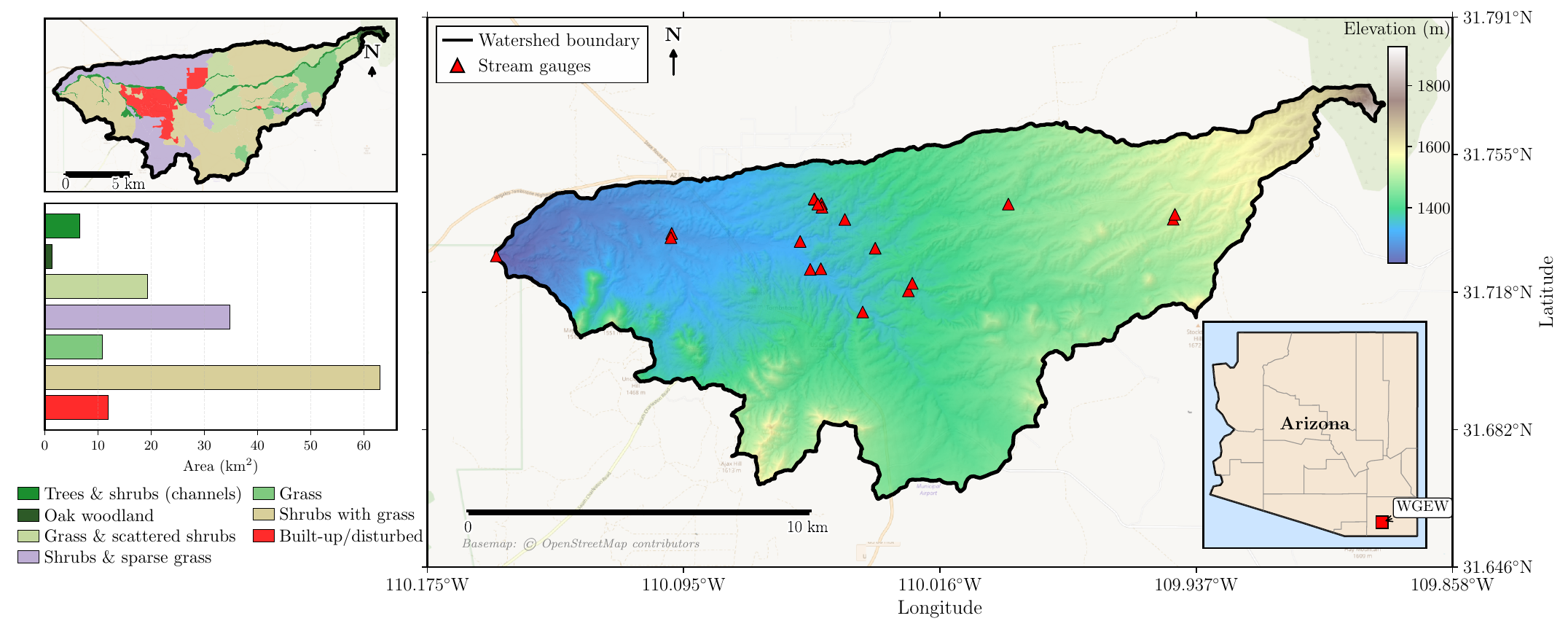}
\vspace{-0.5cm}
\caption{Walnut Gulch Experimental Watershed, Tombstone, Arizona, USA. 
The main panel shows the watershed boundary overlaid on the digital elevation model (DEM), with elevation ranging approximately from $1220\;\mr{m}$ to $1930\; \mr{m}$. Red triangular markers denote the locations of stream gauges within the watershed. The left panels illustrate the land use/land cover (LULC) classification and corresponding area distribution across the 24 vegetation cover types. The dominant classes are shrubs with grass ($63.0\;\mr{km}^2$, $42.7\%$), shrubs and sparse grass ($33.7\;\mr{km}^2$, $22.8\%$), and grass with scattered shrubs ($20.4\;\mr{km}^2$, $13.8\%$), together accounting for approximately $79\%$ of the watershed. Minor classes include built up and disturbed areas ($11.9\;\mr{km}^2$), grass ($10.8\;\mr{km}^2$), woody riparian cover along channels ($6.1\;\mr{km}^2$), and upland oak and juniper woodland ($1.8\;\mr{km}^2$). The inset map (bottom right) indicates the geographical location of WGEW within Arizona, USA.}
\label{fig:walnut_gulch_spatial_overview}
\vspace{-0.5cm}
\end{figure*}

The catchment domain $\mc{D}_{\mr{vtilted}}\subset\mbb{R}^2$ is represented by the cell graph $\mc{G}=(\mc{K}_{\mr{vtilted}},\mc{E})$, where the cell set $\mc{K}_{\mr{vtilted}}=\{(i,j):1\le i\le 81,\;1\le j\le 50\}$ has cardinality $K=4{,}050$ (cell size $\Delta x=\Delta y=20\;\mr{m}$). Under the discrete-time DAE formulation~\eqref{eq:bdf_descriptor}, the state vector $\m x=[\m x_d^\top,\m x_a^\top]^\top\in\Rn{n_x}$ collects the differential states $\m x_d=[\m h^\top,\m F^\top]^\top\in\Rn{n_d}$ (water depth and cumulative infiltration, $n_d=2K$) and the algebraic states $\m x_a=[\m Q_{\mr{L}}^\top,\m Q_{\mr{R}}^\top,\m Q_{\mr{U}}^\top,\m Q_{\mr{D}}^\top]^\top\in\Rn{n_a}$ (directional discharges, $n_a=4K$), yielding $n_x=n_d+n_a=6K=24{,}300$ coupled states over $\mc{K}_{\mr{vtilted}}$. Discharge gauges are placed on the hillslope cells at every $100\;\mr{m}$ in each lateral direction, with no gauges in the central channel, so that $\gamma_{Q_d}^{i,j}=1$ at the gauged hillslope cells and $\gamma_{Q_d}^{i,j}=0$ elsewhere; this defines the measurement matrix $\m C$ and output vector $\m y_k$ in~\eqref{eq:output}.

The outlet boundary condition is set to normal depth with longitudinal bed slope $S_y=0.02$ at the outlet cells $\mc{K}_{\mr{out}}\subseteq\mc{K}_{\mr{vtilted}}$~\cite{Gomes2023}.
Initial conditions are dry, with $h^{i,j}=0$ and $F^{i,j}=0$ for all $(i,j)\in\mc{K}_{\mr{vtilted}}$, and no imposed inflow boundary. The Green-Ampt setup used in this synthetic benchmark follows $K_s=15\;\mr{mm\,h^{-1}}$, $\psi_f=0\;\mr{mm}$, $\theta_i=0.15$, and $\theta_s=0.30$. For numerical simulation, the V-Tilted case uses BDF order one with $\Delta t=1\;\mr{s}$, Newton tolerance $10^{-10}$, over a $1.5\;\mr{h}$ horizon. These conditions are consistent with the benchmark use of V-Tilted in the aforementioned studies, and they provide a controlled scenario to evaluate the agreement between the DAE solution and independent solvers. The DAE solution is validated against HydroPol2D's explicit CA routing and the local inertia formulation~\cite{Gomes2024a}, which are solvers that have been previously validated against analytical solutions and experimental data for 2D coupled overland flow and infiltration modeling. The V-Tilted case serves as a controlled reference scenario to evaluate solver agreement and uncertainty propagation under minimal heterogeneity, while the Walnut Gulch case provides a realistic benchmark for assessing scalability and performance under distributed and complex watershed conditions.

\subsection{Walnut Gulch Experimental Watershed}\label{subsec:walnut_gulch_case_study}
The second case study is the WGEW, shown in Fig.~\ref{fig:walnut_gulch_spatial_overview}, located in southeastern Arizona, USA. The WGEW is a semiarid basin covering approximately $150\;\mr{km}^2$ and is managed by the USDA-ARS Southwest Watershed Research Center~\cite{Nichols2008}. The watershed is characterized by diverse land cover, with shrubland and grassland accounting for approximately $65\%$ of the basin area, as shown in Fig.~\ref{fig:walnut_gulch_spatial_overview}; permanent vegetation transects and a spatially distributed land cover map are available from long term monitoring records~\cite{Skirvin2008}. Soils are spatially heterogeneous, with hydraulic properties varying across the watershed as a function of land cover and subsurface texture~\cite{Nichols2008}. 

The catchment domain $\mc{D}_{\mr{wg}}\subset\mbb{R}^2$ is represented by the cell graph $\mc{G}=(\mc{K}_{\mr{wg}},\mc{E})$, where the cell set $\mc{K}_{\mr{wg}}=\{(i,j):1\le i\le 242,\;1\le j\le 521\}$ has cardinality $K=126{,}082$ (cell size $\Delta x\approx50\;\mr{m}$, $\Delta y\approx52\;\mr{m}$, covering $149\;\mr{km}^2$). Under the discrete-time DAE formulation~\eqref{eq:bdf_descriptor}, we obtain $n_x=n_d+n_a=6K=756{,}492$ coupled states over $\mc{K}_{\mr{wg}}$. All geospatial inputs are obtained from the USDA-ARS SWRC DAP archive~\cite{Nichols2008}. The terrain model is based on a USGS 10 m DEM, resampled to the model grid and conditioned by including the SWRC streamline network to enforce hydrological connectivity; the watershed outlet is fixed at Flume~1. The LULC is derived from the SWRC vegetation polygon map, which classifies the basin into 24 cover types~\cite{Skirvin2008}; each cover type is rasterized to the model grid and assigned a Manning roughness value $n$ per the SWRC vegetation classification~\cite{Skirvin2008}, resulting in a $5\times$ spatial variation in $n$ across the watershed ($n=0.015$ to $0.078\;\mr{s\,m^{-1/3}}$). Soil hydraulic parameters ($K_s$, $\psi_f$, $\theta_s$) are obtained from SSURGO polygon data available through the same archive~\cite{Nichols2008}. \comment{The SSURGO based assignment of $K_s$, $\psi_f$, and $\theta_s$ across Walnut Gulch encodes the spatial heterogeneity of the soil hydraulic parameters used in the Green-Ampt model.}

The parameter ranges used to construct the spatial fields in the DAE simulations are reported in Tab.~\ref{tab:walnut_gulch_params}; it defines the nominal fields for the deterministic simulation and the baseline values for the uncertain parameters $n$, $K_s$, $\psi_f$, and $\eta$ introduced in Section~\ref{subsec:pse}. \comment{We note that the uniform value $\theta_i=0.15$ in Tab.~\ref{tab:walnut_gulch_params} sets the nominal initial moisture state used to construct the operating trajectory $\m x_k^0$; uncertainty in $\theta_i$ is propagated through the uncertain moisture deficit $\eta^{i,j}=\theta_s^{i,j}-\theta_i^{i,j}$ listed in Tab.~\ref{tab:pse_uncertainty}, since $\partial\eta/\partial\theta_i=-1$.}
The spatial heterogeneity of Manning roughness reflects the land cover classes shown in Fig.~\ref{fig:walnut_gulch_spatial_overview}. The corresponding watershed composition is detailed in the same figure. The DAE simulation uses BDF order one with $\Delta t=1\;\mr{s}$, Newton tolerance $10^{-5}$, over a simulation horizon of $90\;\mr{min}$. The DAE solution is validated against HydroPol2D's explicit CA routing and the local inertia formulation~\cite{Gomes2024a}; this case study evaluates the scalability of the proposed framework to $n_x=756{,}492$ differential and algebraic state variables under realistic distributed heterogeneous conditions. The watershed datasets are available from the USDA-ARS Southwest Watershed Research Center through the DAP web interface at \url{http://www.tucson.ars.ag.gov/dap/}.

\begin{table}[t]
\vspace*{0.15cm}
\fontsize{9}{9}\selectfont
\centering
\caption{WGEW parameter ranges used in the DAE simulations. Vegetation layers are assembled from WGEW processed records and vegetation studies~\cite{Skirvin2008}. Soil hydraulic parameters are assigned from SSURGO based soil mapping and USDA ARS SWRC records used in the WGEW data assembly~\cite{Nichols2008}.}
\label{tab:walnut_gulch_params}
\renewcommand{\arraystretch}{1.1}
\resizebox{\linewidth}{!}{%
\begin{tabular}{@{}l p{0.28\linewidth} l p{0.34\linewidth}@{}}
\toprule
Symbol & Description & Unit & Range \\
\midrule
$n$ & Manning roughness & $\mr{s\,m^{-1/3}}$ & $0.015$-$0.078$ \\
$K_s$ & Saturated hydraulic conductivity & $\mr{m\,s^{-1}}$ & $3.67\times 10^{-6}$-$5.83\times 10^{-5}$ \\
$\psi_f$ & Wetting front suction head & $\mr{m}$ & $0.049$-$0.110$ \\
$\theta_s$ & Saturated water content & $[-]$ & $0.437$-$0.463$ \\
$\theta_i$ & Initial water content & $[-]$ & $0.150$ (uniform) \\
$\eta=\theta_s-\theta_i$ & Moisture deficit & $[-]$ & $0.287$-$0.313$ \\
\toprule \bottomrule
\end{tabular}}
\end{table}

\section{Results and discussion} \label{sec:results}
In this section, we evaluate the 2D overland flow and infiltration DAE model and the proposed uncertainty framework on the two watersheds. All simulations are performed in MATLAB R2024b on a MacBook Pro with an Apple M1 Pro chip, a 10-core CPU, and 16 GB RAM. The objectives of the case studies are outlined as follows.
\vspace{-0.2cm}
\begin{itemize}[leftmargin=*]
\item To assess whether the proposed DAE formulation produces a hydrologic and hydrodynamic response consistent with that predicted by independent baseline solvers, specifically the explicit CA routing and local inertia formulations discussed in Section~\ref{subsec:dae} for both the synthetic V-Tilted and the real WGEW watersheds.\vspace{-0.3cm}
\item To quantify how uncertainty in rainfall intensity and watershed parameters propagates to water depth, infiltration, and directional discharge states through the proposed uncertainty framework in Section~\ref{subsec:pse}.\vspace{-0.3cm}
\item To characterize the relative influence of uncertainty sources across space and time and to evaluate their impact on the confidence intervals of hydraulic state variables in measured and unmeasured locations.\vspace{-0.3cm}
\item To evaluate the computational scalability of the proposed framework for large heterogeneous domains under time-varying rainfall forcing, using the WGEW as a realistic benchmark.
\end{itemize}

\begin{figure*}[h]
\centering
\subfloat{\includegraphics[keepaspectratio=true,width=\columnwidth]{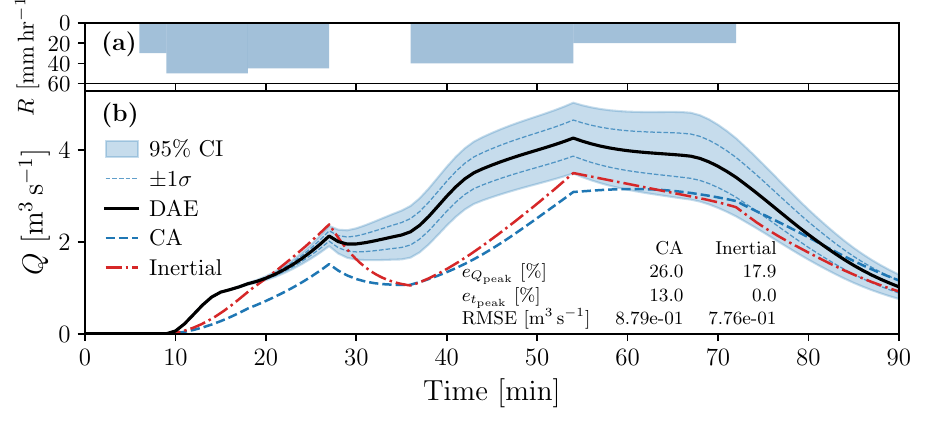}\label{fig:solver_comparison_vtilted}}{}{}
\subfloat{\includegraphics[keepaspectratio=true,width=\columnwidth]{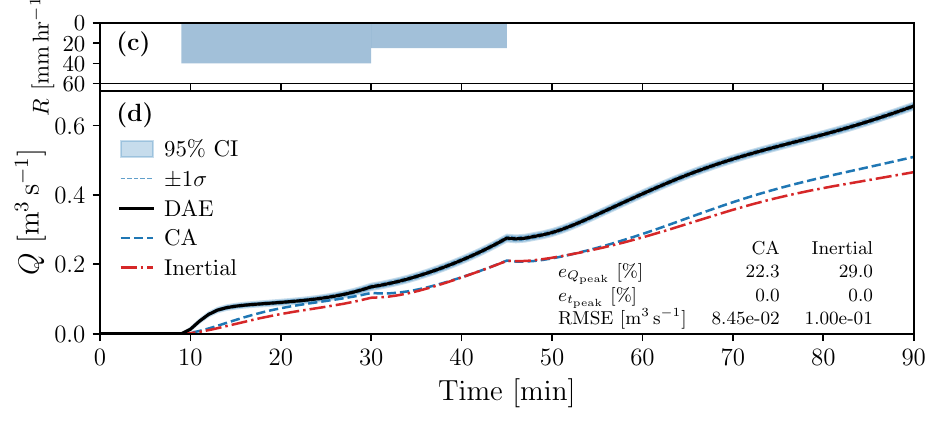}\label{fig:solver_comparison_walnut_gulch}}{}{}
\caption{Hydrograph comparison for the V-Tilted and Walnut Gulch watersheds under the corresponding rainfall patterns $\m R(k)$ [$\mr{mm\,h^{-1}}$]. \comment{The comparison shows numerical agreement between the proposed DAE formulation and the explicit CA and semi-implicit local inertia baselines under the same governing equations; differences in peak magnitude follow from the simultaneous Newton solve in~\eqref{eq:newton_iteration}.} (a) V-Tilted rainfall pattern $\m R(k)$ and (b) the respective outlet discharge trajectories from the proposed DAE formulation, explicit CA routing, and local inertia solvers. (c) Walnut Gulch rainfall pattern $\m R(k)$ and (d) the respective outlet discharge trajectories from the proposed DAE formulation, explicit CA routing, and local inertia solvers.}\label{fig:solver_comparison}
\vspace{-0.3cm}
\end{figure*}
\subsection{DAE formulation response and computational scalability}\label{subsec:results_solver}

\begin{figure*}[t]
\centering
\includegraphics[keepaspectratio=true,width=\textwidth]{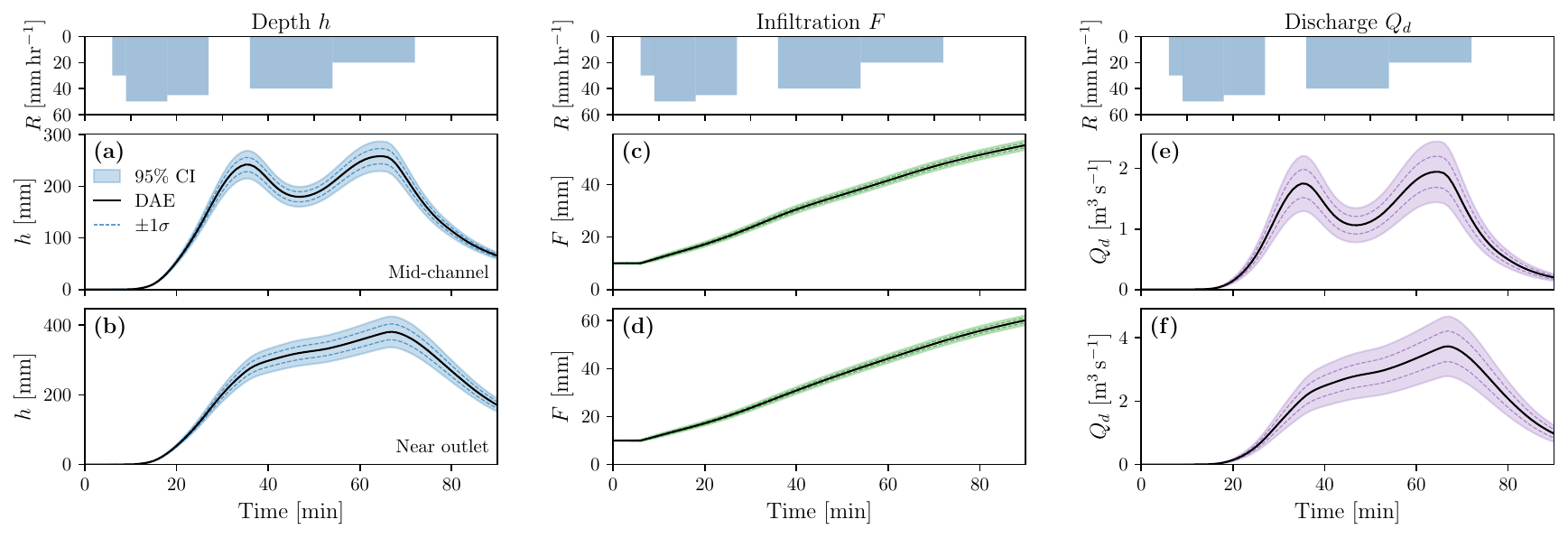}
\caption{Time domain uncertainty envelopes for the V-Tilted case study at the mid-channel and near-outlet cells. The figure shows the nominal DAE trajectory~$\m x_k^0$, the $\pm1\sigma$ band, and the $95\%$ confidence interval for water depth~$h$ [mm], cumulative infiltration~$F$ [mm], and directional discharge~$Q_d$ [m$^3$\,s$^{-1}$] at each location, with the corresponding event rainfall~$R(k)$ [mm\,h$^{-1}$] shown in each column. Confidence interval width grows during the active runoff period and narrows during recession, consistent with the covariance mapping in~\eqref{eq:Kxx}, with wider envelopes at the near-outlet cell reflecting upstream uncertainty accumulation through the coupled routing dynamics in~\eqref{eq:diff_states} and~\eqref{eq:algebraic_constraint}.}\label{fig:pse_envelopes_vtilted}
\end{figure*}

\begin{figure*}[h]
\centering
\includegraphics[keepaspectratio=true,width=\textwidth]{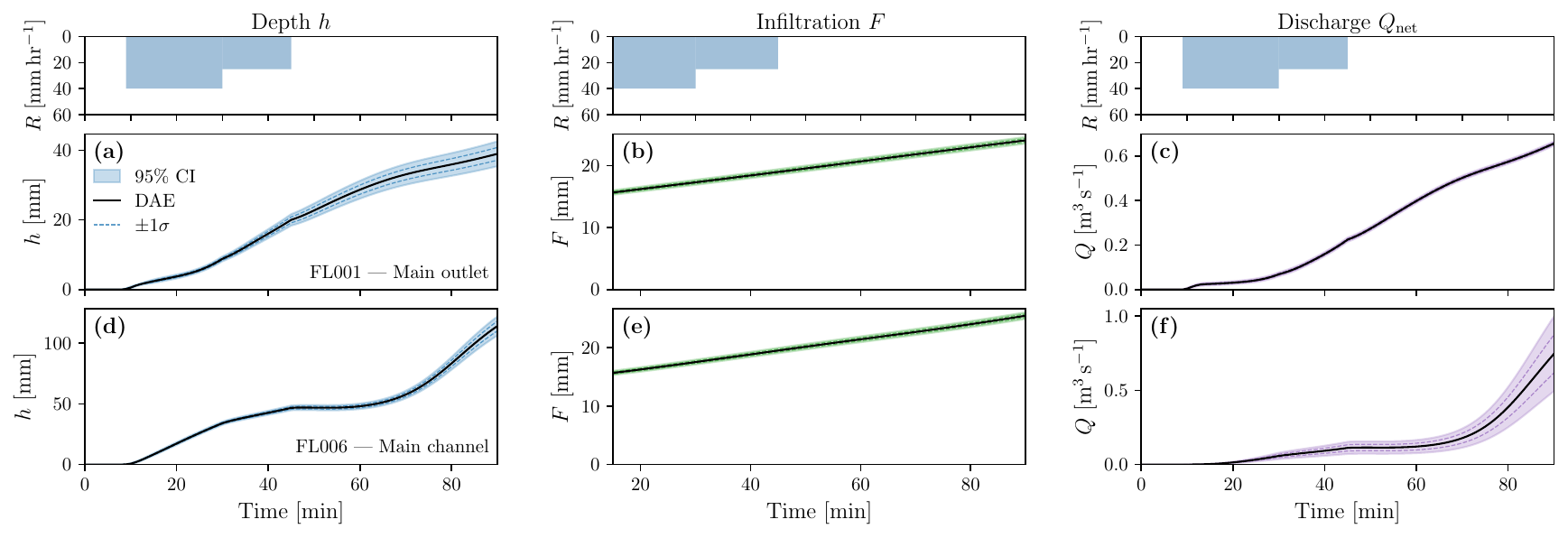}
\caption{Time domain uncertainty envelopes for Walnut Gulch at FL001 (main outlet) and FL006 (main channel). The figure depicts the nominal DAE trajectory~$\m x_k^0$, the $\pm1\sigma$ band, and the $95\%$ confidence interval for water depth~$h$ [mm], cumulative infiltration~$F$ [mm], and directional discharge~$Q_d$ [m$^3$\,s$^{-1}$] at each location, with the associated event rainfall~$R(k)$ [mm\,h$^{-1}$] shown in each column. The covariance propagation framework in~\eqref{eq:Kxx} produces uncertainty envelopes at both gauged locations under the distributed rainfall and heterogeneous parameter fields of the real watershed, with interval widths that evolve in response to the transient runoff dynamics resolved through the coupled DAE structure in Section~\ref{subsec:dae}.}\label{fig:pse_envelopes_walnut_gulch}
\end{figure*}

The first objective of the case studies is to assess whether the proposed semi-explicit DAE formulation, introduced in Section~\ref{subsec:dae}, produces a hydrologic and hydrodynamic response consistent with established baseline solvers. Outlet hydrograph comparisons are carried out against HydroPol2D's explicit CA routing~\cite{Gomes2023} and local inertia~\cite{Gomes2024a} formulations for both watersheds. Response validity is quantified using the relative peak flow error, the time-to-peak error, and the outlet discharge RMSE over horizon $N_T$. The relative peak flow error is defined as the absolute difference between the DAE predicted peak discharge $Q_{\mr{peak}}^{\mr{DAE}}$ and the reference peak discharge $Q_{\mr{peak}}^{\mr{ref}}$ from the baseline solver, normalized by the maximum of the absolute reference peak discharge and a small regularization constant $\varepsilon_Q$ to avoid division by zero. The time-to-peak error is defined similarly, comparing the DAE predicted time to peak $t_{\mr{peak}}^{\mr{DAE}}$ with the reference time to peak $t_{\mr{peak}}^{\mr{ref}}$, normalized by the maximum of the reference time to peak and a regularization constant $\varepsilon_t$. The measures can be written as follows.
\begin{equation}\label{eq:metric_peak}
e_{Q_{\mr{peak}}}=\frac{\left|Q_{\mr{peak}}^{\mr{DAE}}-Q_{\mr{peak}}^{\mr{ref}}\right|}{\max\!\left(\left|Q_{\mr{peak}}^{\mr{ref}}\right|,\varepsilon_Q\right)},\qquad
e_{t_{\mr{peak}}}=\frac{\left|t_{\mr{peak}}^{\mr{DAE}}-t_{\mr{peak}}^{\mr{ref}}\right|}{\max\!\left(t_{\mr{peak}}^{\mr{ref}},\varepsilon_t\right)},
\end{equation}
and the outlet discharge RMSE over horizon $N_T$,
\begin{equation}\label{eq:metric_rmse}
\mr{RMSE}_{\m Q}=\sqrt{\frac{1}{N_T}\sum_{k=1}^{N_T}\left\|\m Q_{\mr{DAE}}(k)-\m Q_{\mr{ref}}(k)\right\|_2^2}.
\end{equation}

We note that the Nash-Sutcliffe Efficiency (NSE) and Kling-Gupta Efficiency (KGE) \comment{are computed alongside the absolute error metrics in this work. As discussed in~\cite{Williams2025}, NSE and KGE are relative metrics whose values depend on the variability and temporal correlation of the reference hydrograph rather than only on the absolute prediction error; the reply in~\cite{Clark2026} notes that, when reported together with absolute error metrics, NSE and KGE decompose the solver to reference mismatch into correlation, variability, and bias components $(r,\alpha,\beta)$.} Solver agreement is assayed using $e_{Q_{\mr{peak}}}$, $e_{t_{\mr{peak}}}$, $\mr{RMSE}_{\m Q}$\comment{, NSE, and KGE}, where $e_{Q_{\mr{peak}}}$ and $e_{t_{\mr{peak}}}$ characterize peak and timing mismatch in normalized terms, while $\mr{RMSE}_{\m Q}$ quantifies outlet discharge trajectory mismatch in physical units. For uncertainty quantification, we report mean trajectories, variance envelopes, and confidence intervals extracted from $\mr{K}_{\m x\m x}(k)$ as explained in Section~\ref{subsec:pse}.
In this subsection, hydrograph agreement is evaluated using deterministic outlet discharge trajectories from the DAE formulation against deterministic CA routing and local inertia baseline trajectories, without covariance bounds.

Fig.~\ref{fig:solver_comparison_vtilted} shows the V-Tilted outlet hydrograph comparison. The local inertia baseline recovers peak timing relative to the DAE trajectory ($e_{t_{\mr{peak}}}=0.0\%$), while both baseline solvers show peak magnitude differences ($e_{Q_{\mr{peak}}}=17.9\%$ for local inertia and $e_{Q_{\mr{peak}}}=26.0\%$ for CA); the CA baseline also shows a peak timing difference ($e_{t_{\mr{peak}}}=13.0\%$). The RMSE values are $7.76\times 10^{-1}\;\mr{m^3\,s^{-1}}$ for local inertia and $8.79\times 10^{-1}\;\mr{m^3\,s^{-1}}$ for CA, consistent with the expected behavior in a controlled synthetic setting with minimal geometric and parameter heterogeneity. \comment{For this case, the CA routing baseline yields $\mr{NSE}=0.597$ and $\mr{KGE}=0.622$ (with correlation $r=0.924$, variability ratio $\alpha=0.772$, and bias ratio $\beta=0.709$), while the local inertia baseline yields $\mr{NSE}=0.686$ and $\mr{KGE}=0.659$ (with $r=0.939$, $\alpha=0.775$, $\beta=0.751$).} The higher peak flow in the DAE trajectory relative to both baselines is attributed to the simultaneous solution of routing and infiltration within each time step. In the explicit CA scheme and the semi-implicit local inertia formulation, infiltration and Manning routing are updated through separate operations within the time step: infiltration removes water from ponded depth before lateral concentration from upstream cells, while routing redistributes flow without a concurrent infiltration loss in the same coupled solve. This sequential decoupling can attenuate outlet peak discharge. However, in the DAE formulation, the mass balance~\eqref{eq:depth_update}, infiltration~\eqref{eq:infil_update}, and algebraic Manning constraints~\eqref{eq:algebraic_constraint} are enforced simultaneously through the Newton solve~\eqref{eq:newton_iteration}, so the infiltration rate $\m f(t)$ at each cell is evaluated at a water depth that satisfies the lateral routing balance and the vertical infiltration loss within the same step, resulting in lower numerical attenuation and higher peak concentration at the outlet. \comment{This result follows from solving the differential and algebraic constraints simultaneously in the Newton iteration~\eqref{eq:newton_iteration} and characterizes the agreement between the DAE formulation and the baseline solvers under the same governing equations; it is not an indication that the DAE trajectory is closer to the true catchment response, which would require comparison against field measured discharge.} This result validates the DAE structure established in Proposition~\ref{thm:index1_local} and confirms that the BDF implicit discretization in~\eqref{eq:bdf_descriptor} accurately preserves the mass balance, Manning routing, and Green-Ampt infiltration interactions formulated in Section~\ref{subsec:hydropol2d}.

Fig.~\ref{fig:solver_comparison_walnut_gulch} shows the Walnut Gulch outlet hydrograph comparison. Both baselines recover peak timing ($e_{t_{\mr{peak}}} = 0.0\%$), while peak magnitude differences are larger than in the synthetic case ($e_{Q_{\mr{peak}}} = 20.8\%$ for CA and $e_{Q_{\mr{peak}}} = 27.5\%$ for local inertia), reflecting the influence of distributed topographic, vegetation, and soil heterogeneity on routing approximations across the $149\;\mr{km}^2$ basin. The same sequential explicit and semi-implicit numerical treatment used in the V-Tilted case contributes to the higher DAE peak. Under distributed conditions, the discrepancy is larger because the nonlinear interaction between spatially variable infiltration capacity ($K_s$, $\psi_f$) and Manning roughness is accounted for jointly in the DAE Newton step~\eqref{eq:newton_iteration}, while the explicit and semi-implicit simulations decouple these interactions within each time step. The RMSE values remain moderate ($7.68\times 10^{-2}\;\mr{m^3\,s^{-1}}$ for CA and $9.24\times 10^{-2}\;\mr{m^3\,s^{-1}}$ for local inertia), indicating that the DAE descriptor structure in~\eqref{eq:semi_NDAE_rep} remains consistent with established solvers under realistic heterogeneous watershed conditions. \comment{Furthermore, the CA routing baseline yields $\mr{NSE}=0.835$ and $\mr{KGE}=0.668$ (with $r=0.997$, $\alpha=0.773$, $\beta=0.759$), while the local inertia baseline yields $\mr{NSE}=0.769$ and $\mr{KGE}=0.605$ (with $r=0.998$, $\alpha=0.729$, $\beta=0.714$). The high correlation across the watersheds indicates that the baselines preserve the timing and shape of the DAE outlet hydrograph, while the $\alpha,\beta<1$ ratios are consistent with the lower peak magnitude and lower mean discharge of the baseline solvers relative to the DAE trajectory.}

The above case studies illustrate the scalability of the proposed DAE simulation. For the V-Tilted case, the state dimension is $n_x=24{,}300$ states, while for Walnut Gulch it is $n_x=756{,}492$ states. This corresponds to an increase in state dimension by a factor of approximately $31.1$ from the synthetic to the real watershed. The solver runtimes for the implicit BDF method are $22.34\;\mr{s}$ for V-Tilted and $491.90\;\mr{s}$ for Walnut Gulch, which reflects a moderate increase in runtime relative to the increase in state dimension. This scalability is attributed to the sparsity of the Jacobian matrix $\m J_{k+1}^{(i)}$ used in the Newton iteration, which exploits the local connectivity structure of the cell graph $\mc{G}$. \comment{The stacked system matrix $\m A[k]$ in the covariance mapping~\eqref{eq:Kxx} inherits this sparsity. The Jacobian $\m A_k^{\m\psi}$ couples a central cell to at most its four von Neumann neighbors, thus the number of nonzero entries per row is bounded by a small constant independent of $K$, and the total number of nonzero entries in $\m A[k]$ grows linearly with $K$. The block structure in~\eqref{eq:Ak_structure} enables the covariance computation to be decomposed and evaluated independently at each cell without assembling the global stacked matrix $\m A[k]$ explicitly.} However, it is important to note that the DAE BDF solver requires more computational effort per time step than the explicit CA routing and local inertia baselines, which have runtimes of $3.22\;\mr{s}$ and $10.81\;\mr{s}$ for V-Tilted, and $56.01\;\mr{s}$ and $195.59\;\mr{s}$ for Walnut Gulch, respectively. While the computational cost of the DAE solver is higher, the implicit Newton solve in~\eqref{eq:newton_iteration} enforces the coupled differential dynamics~\eqref{eq:diff_states} and the algebraic Manning discharge constraints~\eqref{eq:algebraic_constraint} simultaneously at each time step, rather than decoupling them as in the explicit/semi-implicit baselines. This simultaneous coupling is essential for the covariance propagation framework: the linearized DAE Jacobians $\m A_k^{\m\psi}$ and $\m B_{\theta,k}$ in~\eqref{eq:Ak_Btheta_def} are evaluated at operating points where the mass balance, infiltration, and flow constraints are jointly satisfied, ensuring that the propagated uncertainty in $\mr{K}_{\m x\m x}(k)$ reflects the full coupled system response. Uncertainty quantification is then evaluated in the remaining subsections using the covariance mapping in Proposition~\ref{thm:pse_solution} under the regularity and solvability conditions of Assumption~\ref{asmp:dae_regular}.

\subsection{Uncertainty propagation and confidence intervals}\label{subsec:results_uncertainty}
To address the second and third objectives of the case studies, uncertainty in rainfall forcing and hydraulic parameters $\{R, n, K_s, \psi_f, \eta\}$ is propagated simultaneously through the closed-form covariance mapping in~\eqref{eq:Kxx} for the distributions considered in Tab.~\ref{tab:pse_uncertainty}. In the numerical setup, the covariance model uses the following coefficients of variation for the parameters $(\mr{CV}_R, \mr{CV}_n, \mr{CV}_{K_s}, \mr{CV}_{\psi}, \mr{CV}_{\eta}) = (0.25, 0.20, 0.25, 0.10, 0.12)$, and the variances for the measurement noise are defined as $\mr{Var}(v_h) = 2.5 \times 10^{-7}\;\mr{m^2}$, $\mr{Var}(v_Q) = 1.0 \times 10^{-8}\;\mr{m^6\,s^{-2}}$, and $\mr{Var}(v_F) = 2.5 \times 10^{-5}\;\mr{m^2}$. \comment{Note that the confidence intervals are obtained under Assumption~\ref{asmp:independence}, which models the uncertain inputs and parameters as mutually independent. In a real watershed, positive spatial correlation in $K_s$, $\psi_f$, $\eta$, $n$, or rainfall would amplify the envelope width along drainage corridors, where correlated deviations in upstream cells propagate and amplify toward the outlet, while negative correlation in the same inputs would narrow it; the magnitude of this effect at gauged and ungauged locations under the resulting covariance matrices is not quantified in the present case studies.}

\begin{figure*}[t]
\centering
\includegraphics[keepaspectratio=true,width=0.98\textwidth]{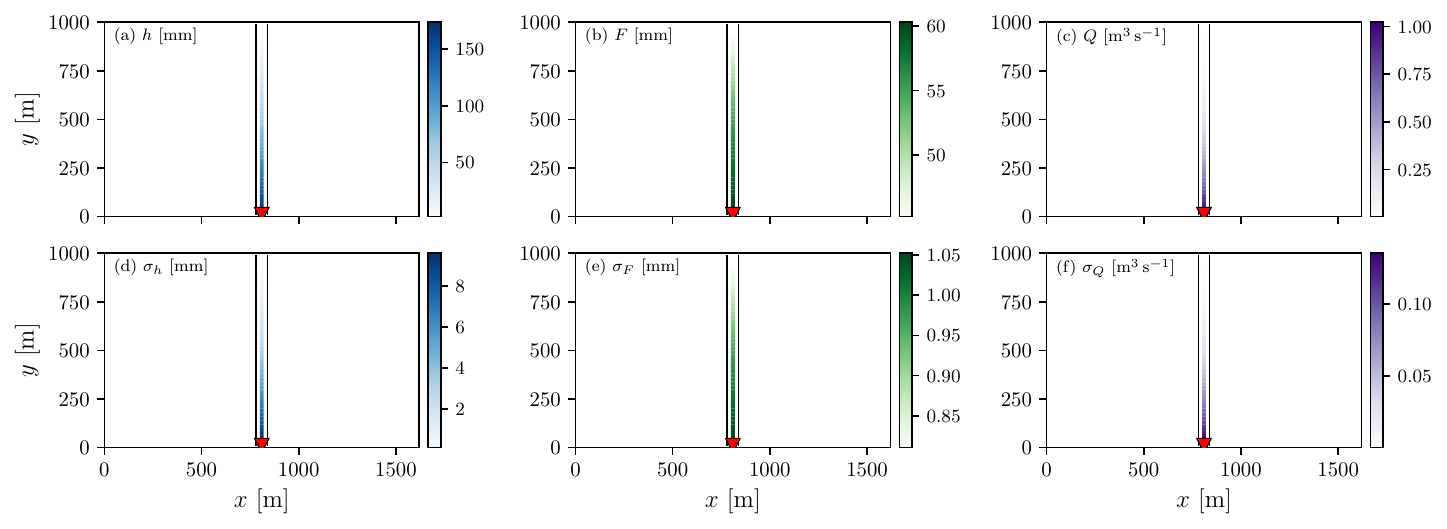}
\caption{Spatial uncertainty response for the V-Tilted catchment; the figure shows representative spatial snapshots at time $t=90\;\mr{min}$ for $h$, $F$, and $Q$, and the corresponding propagated uncertainties $\sigma_h$, $\sigma_F$, and $\sigma_Q$. The results show that uncertainty is not spatially uniform, with larger values along dominant routing paths and near the outlet.}
\label{fig:spatial_uncertainty_vtilted}
\end{figure*}

\begin{figure*}[t]
\centering
\includegraphics[keepaspectratio=true,width=\textwidth]{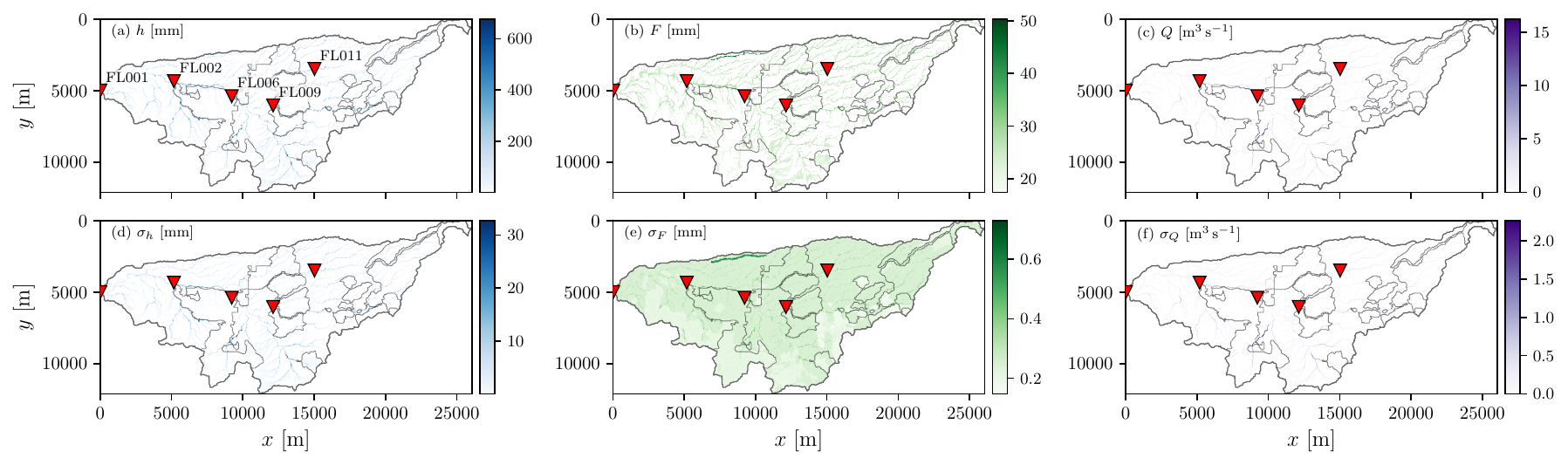}
\caption{Spatial uncertainty response for the Walnut Gulch watershed; the figure shows representative spatial snapshots at time $t=90\;\mr{min}$ for $h$, $F$, and $Q$, and the corresponding propagated uncertainties $\sigma_h$, $\sigma_F$, and $\sigma_Q$ over the heterogeneous basin. The results illustrate how topography, drainage structure, and rainfall pattern jointly shape spatial variance in predicted states. Relative to synthetic settings, the heterogeneous watershed exhibits stronger localization of high uncertainty regions along active channel pathways.}\label{fig:spatial_uncertainty_walnut_gulch}
\vspace{-0.2cm}
\end{figure*}

Fig.~\ref{fig:pse_envelopes_vtilted} and Fig.~\ref{fig:pse_envelopes_walnut_gulch} show the uncertainty envelopes for watershed state variables $\m h$, $\m F$, and $\m Q_d$ at selected cells. In both case studies, the confidence interval width increases during active runoff periods and decreases during recession. This behavior follows directly from the covariance mapping given in~\eqref{eq:Kxx}, where the state covariance $\mr{K}_{\m x\m x}(k)$ is driven by $\mr{K}_{\m b\m b}[k]$ through the stacked linearized DAE and measurement model assembled in~\eqref{eq:Ak_structure}. This means that the confidence intervals computed around the nominal trajectory $\m x_k^0$ via~\eqref{eq:normal_propagation} quantify uncertainty growth and decay relative to the event dynamics. At measured locations, the stacked system $\m A[k]\m x[k]=\m b[k]$ is overdetermined because $n_s=n_x+n_y>n_x$ equations are available for $n_x$ states, as discussed in Section~\ref{subsubsec:cov_prop}; consequently, the covariance solution in~\eqref{eq:Kxx} leverages the additional measurement equations to reduce state variance, producing tighter confidence intervals at gauged cells relative to ungauged locations where only the $n_x$ DAE equations constrain the solution.

The proposed method computes the uncertainty response for both measured and unmeasured states from the closed-form covariance map in~\eqref{eq:Kxx} under partial measurements, with uncertainty bounds defined by the linearization around the nominal trajectory $\m x_k^0$ as stated in Proposition~\ref{thm:pse_solution}. The spatial plots depicting the state variables and their corresponding propagated uncertainties at the final time step of the simulation are shown in Fig.~\ref{fig:spatial_uncertainty_vtilted} and Fig.~\ref{fig:spatial_uncertainty_walnut_gulch}. Notice that the propagated uncertainty is not uniformly distributed over the domain. For V-Tilted, the largest spatial variance for $\m h$ and $\m Q_d$ concentrates along dominant routing paths and near the outlet, where cumulative flow contributions from upstream cells are highest; this follows from the coupled differential and algebraic structure in~\eqref{eq:diff_states}-\eqref{eq:algebraic_constraint}, where overland flow routing and infiltration are solved simultaneously via the implicit Newton step~\eqref{eq:newton_iteration}, and upstream uncertainty accumulates through the directional discharge interactions defined by~\eqref{eq:inflow}. For WGEW, high variance regions are localized along active drainage and channel segments, while upland and low response areas remain comparatively low variance; this spatial pattern reflects the distributed topographic and drainage structure of the watershed, where flow convergence and routing interactions amplify uncertainty along active pathways while low response areas are less sensitive to rainfall and parameter uncertainty. 

We note here that the covariance propagation runtimes for the presented runs are $11.67\;\mr{s}$ for V-Tilted and $182.60\;\mr{s}$ for WGEW. Although the DAE solver requires more computational time than the explicit CA routing and local inertia baselines, as discussed earlier, the covariance propagation can be evaluated in realtime alongside the DAE trajectory. This additional computational cost is relatively lower than that required to obtain uncertainty estimates through MC ensemble methods, which requires repeated model evaluations (as shown in the subsequent section). Thus, the proposed covariance propagation method provides a computationally efficient approach for realtime uncertainty quantification for 2D overland flow and infiltration simulation over the domain $\mc{D}$, enabling confidence aware estimates of watershed conditions without the need for multiple model runs, which can be limiting for large state dimensions such as in the WGEW case.

\begin{figure}[t]
\centering
\includegraphics[keepaspectratio=true,width=\columnwidth]{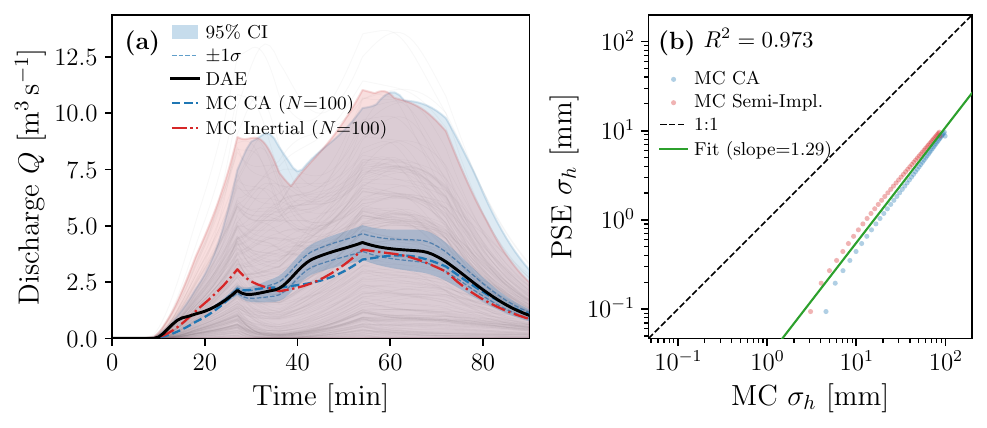}
\caption{MC simulations for the V-Tilted case; (a) compares nominal DAE outlet discharge response and covariance based confidence intervals against MC hydrograph ensembles; (b) compares uncertainty for $\m h$ from covariance propagation and MC ensembles.}\label{fig:mc_validation_vtilted}
\end{figure}

\subsection{MC ensemble validation}\label{subsec:results_mc_validation}
Here, we show that for MC ensemble hydrograph realizations generated using HydroPol2D's CA routing and the semi-implicit formulation~\cite{Gomes2023,Gomes2024a}, the uncertainty bounds obtained from the closed-form covariance mapping in Proposition~\ref{thm:pse_solution} and Algorithm~\ref{algorithm1} are consistent with those obtained from MC ensemble based simulations, from the synthetic benchmark to the large heterogeneous watershed. This comparison directly evaluates whether the covariance solution from Problem~\ref{prob:pse_cov_linear_form} reproduces ensemble based uncertainty statistics at the outlet and over the spatial domain. Fig.~\ref{fig:mc_validation_vtilted}(a) depicts the $100$ MC simulations under the CA routing and semi-implicit baselines for the V-Tilted case, while Fig.~\ref{fig:mc_validation_walnut_gulch}(a) shows the CA routing MC simulations for WGEW. Note that we did not simulate the semi-implicit MC for WGEW due to the computational cost of repeated model evaluations for $N_{\mr{MC}}=100$. Furthermore, we use Latin Hypercube Sampling (LHS) to generate the $100$ realizations for the uncertain parameters and rainfall, ensuring a more efficient exploration of the uncertainty space compared to simple random sampling~\cite{Helton2003}. The MC hydrograph ensembles are compared against the nominal DAE trajectory and the covariance based confidence intervals at the outlet.

The computed covariances from the proposed framework are within the range of the MC ensemble trajectories. In Fig.~\ref{fig:mc_validation_vtilted}(b) and Fig.~\ref{fig:mc_validation_walnut_gulch}(b), the $R^2$ value is computed between covariance propagated variance and MC sample variance, evaluated over all the cells $\mc{K}$ within the watershed domain $\mc{D}$. The $R^2$ quantifies how much of the spatial variance pattern from MC is explained by the covariance mapping, while the fitted line is interpreted jointly through slope and intercept in the log scale relation $\sigma_{\mr{PSE}}=10^b \sigma_{\mr{MC}}^{a}$.

For V-Tilted, Fig.~\ref{fig:mc_validation_vtilted}(b) gives $R^2=0.973$ with $(a,b)=(1.29,-1.55)$, while for Walnut Gulch, Fig.~\ref{fig:mc_validation_walnut_gulch}(b) gives $R^2=0.852$ with $(a,b)=(1.05,-0.96)$. These values indicate strong spatial pattern consistency in both watersheds, while the negative intercept shifts the fitted relation below the line with slope one. This is due to the variance reduction effect from sensing in the proposed framework being a probabilistic state estimation formulation under partial measurements, where the stacked system in~\eqref{eq:Ak_structure} is overdetermined ($n_s=n_x+n_y>n_x$), and the covariance solution in~\eqref{eq:Kxx} is measurement conditioned, analogous to covariance reduction in Kalman based state estimation. By contrast, the MC variance here is obtained from system realizations without such measurement conditioning.

\begin{figure}[t]
\centering
\includegraphics[keepaspectratio=true,width=\columnwidth]{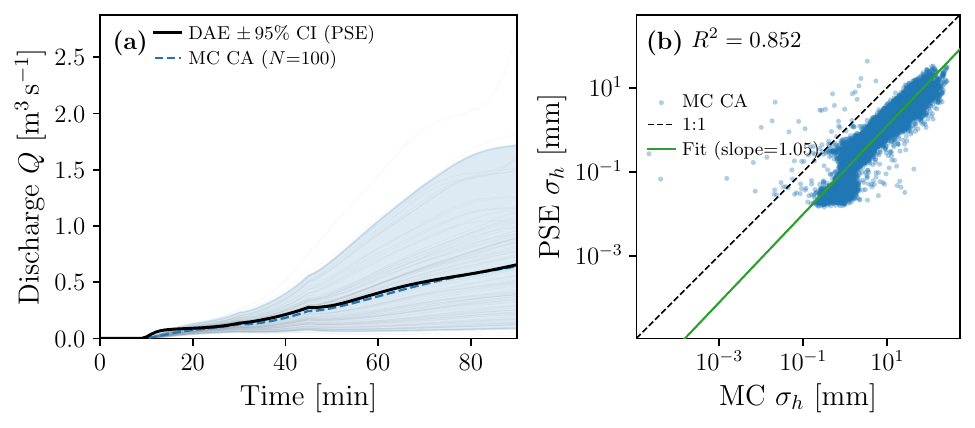}
\caption{MC simulations for Walnut Gulch; (a) compares nominal outlet discharge trajectories and covariance based confidence intervals with MC hydrograph realizations; (b) compares uncertainty for $\m h$ from covariance propagation and ensemble estimates.}\label{fig:mc_validation_walnut_gulch}
\end{figure}

\comment{To assess the MC distributions relative to the covariance based intervals, we report the MC sample skewness $\gamma_1$ and excess kurtosis $\gamma_2$ of the outlet discharge $Q$ at the peak time, along with the empirical coverage, defined as the fraction of MC realizations of $Q$ at the peak that fall within the covariance based confidence interval, against the nominal $95\%$ level. For V-Tilted, the peak occurs at $t=3720\;\mr{s}$ with $\gamma_1=+0.68$ and $\gamma_2=-0.32$, and an empirical coverage of $95.0\%$; for Walnut Gulch, the peak occurs at $t=5400\;\mr{s}$ with $\gamma_1=+1.29$ and $\gamma_2=+2.54$, and an empirical coverage of $95.0\%$. The skewness is positive in both cases and larger for the Walnut Gulch, while the empirical coverage matches the nominal level. This indicates that the covariance based intervals reproduce the MC envelopes, with the residual asymmetry concentrated in the upper tail of the WGEW event.}

\comment{We now compute the empirical coverage of MC realizations of the outlet discharge that fall within the $95\%$ confidence interval $\mu(t)\pm 1.96\,\sigma(t)$ over the rising, peak, and recession phases of the hydrograph, defined relative to the nominal peak time $t_{\mr{peak}}$ and peak discharge $Q_{\mr{peak}}$: rising ($t<t_{\mr{peak}}$, $Q_{\mr{mean}}(t)<0.9\,Q_{\mr{peak}}$), peak ($Q_{\mr{mean}}(t)\geq 0.9\,Q_{\mr{peak}}$), and recession ($t>t_{\mr{peak}}$, $Q_{\mr{mean}}(t)<0.9\,Q_{\mr{peak}}$). For V-Tilted, the coverages are $96.1\%$ (rising), $95.3\%$ (peak), $96.5\%$ (recession), and $96.1\%$ overall; for Walnut Gulch, $95.6\%$ (rising), $94.3\%$ (peak), and $95.5\%$ overall, with the recession coinciding with the end of the simulation window. All values remain within one percentage point of the nominal $95\%$. The lower coverage at the Walnut Gulch peak ($94.3\%$ versus $95.3\%$ at V-Tilted) is consistent with the higher skewness and excess kurtosis of the ensemble and quantifies the linearization bias that a symmetric interval imparts on a non-Gaussian peak distribution. At ungauged cells, the spatial standard deviation comparison in Fig.~\ref{fig:mc_validation_vtilted}(b) and Fig.~\ref{fig:mc_validation_walnut_gulch}(b) provides the corresponding validation. We further evaluate the coverage at V-Tilted cells near the outlet, the considered gauged flumes, and two ungauged interior cells in Walnut Gulch; across all locations, the coverage stays within two and a half percentage points of the nominal $95\%$ on every phase within the simulation window (V-Tilted: $95.5$-$97.2\%$; Walnut Gulch ungauged: $92.6$-$97.9\%$; gauged flumes: $92.7$-$97.0\%$). The peak decrease at the ungauged cells ($92.6$-$92.8\%$) and at FL002 ($92.7\%$) integrates along the drainage network into the $0.7$ percentage point decrease at the outlet ($94.3\%$), consistent with the heavier tailed peak distribution.}

\begin{figure}[t]
\centering
\includegraphics[keepaspectratio=true,width=\columnwidth]{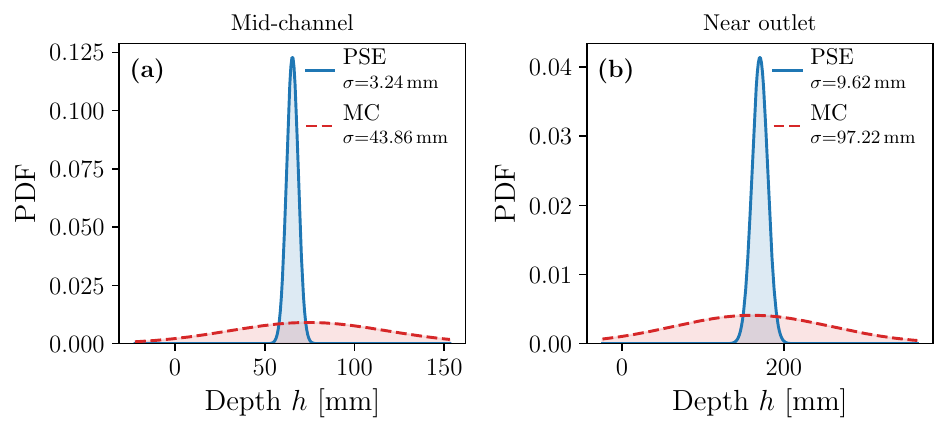}
\caption{Probability densities for depth $h$ in the V-Tilted catchment from covariance propagation and MC distributions from the CA routing; (a) and (b) correspond to two channel locations (mid channel and near outlet).}\label{fig:cell_distributions_vtilted}
\end{figure}

Based on the above results, we assess the resulting probability distribution bounds for gauged and ungauged cells. Fig.~\ref{fig:cell_distributions_vtilted} and Fig.~\ref{fig:cell_distributions_walnut_gulch} provide an additional spatial interpretation regarding the uncertainty response at measured and unmeasured locations across the two watersheds. For V-Tilted, at the two selected ungauged locations within the channel, the proposed covariance propagation gives depth standard deviations of $3.24\;\mr{mm}$ (mid channel) and $9.62\;\mr{mm}$ (near outlet), while the corresponding MC values are $43.86\;\mr{mm}$ and $97.22\;\mr{mm}$. For Walnut Gulch, at the two gauged locations (main outlet and main channel), the proposed covariance propagation gives $1.83\;\mr{mm}$ and $4.04\;\mr{mm}$, while MC gives approximately $15\;\mr{mm}$ and $60\;\mr{mm}$. For the ungauged locations, the corresponding distribution spreads are shown in Fig.~\ref{fig:cell_distributions_walnut_gulch} and remain larger than the gauged cases. In both case studies, the covariance based probability densities at measured locations are narrower than those at unmeasured locations, consistent with the measurement conditioning (overdetermined) discussed in Sections~\ref{subsubsec:cov_prop} and~\ref{subsec:results_uncertainty}. The additional $n_y$ measurement equations appended to the $n_x$ DAE equations in the stacked system~\eqref{eq:Ak_structure} constrain the covariance solution in~\eqref{eq:Kxx}, reducing the posterior variance at gauged cells while the coupled DAE dynamics propagate this information to all the state variables through the DAE state space representation~\eqref{eq:dae_system}.

\begin{figure}[t]
\centering
\includegraphics[keepaspectratio=true,width=\columnwidth]{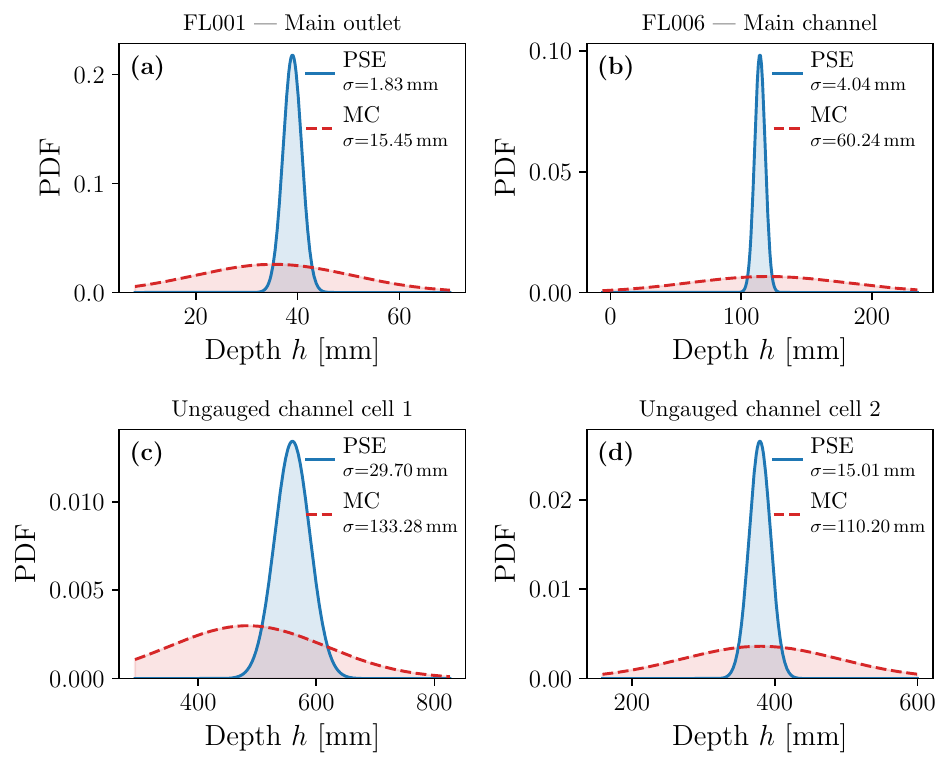}
\caption{Probability densities for depth $h$ in Walnut Gulch from covariance propagation and MC distributions from the CA routing; (a) and (b) correspond to the two gauged locations (main outlet and main channel), while (c) and (d) correspond to two ungauged locations.}\label{fig:cell_distributions_walnut_gulch}
\end{figure}

With that in mind, the covariance propagation method reproduces the main ensemble uncertainty patterns while avoiding repeated model evaluations. The MC simulations are executed in MATLAB using 10 parallel processes. Under this setup, the MC runtimes for $N_{\mr{MC}}=100$ realizations are $116.29\;\mr{s}$ (CA routing) and $394.23\;\mr{s}$ (semi-implicit) for V-Tilted, and $1668.22\;\mr{s}$ (CA routing) for Walnut Gulch. Compared to the covariance propagation runtimes reported in Section~\ref{subsec:results_uncertainty}, the proposed framework is approximately ten times faster than the MC simulations for both watersheds, while requiring only a single model run to obtain the nominal trajectory $\m x_k^0$ and the covariance mapping in~\eqref{eq:Kxx}. From a practical perspective, these results indicate that the proposed framework can provide confidence aware watershed state estimates at gauged and ungauged locations in realtime, which is relevant for operational scenarios where ensemble based uncertainty quantification exceeds the available forecast window and where ungauged locations require reliable uncertainty bounds for decision support in flood warning and water resources management. This concludes the results section.

\section{Conclusion, paper limitations, and recommendations for future work}\label{sec:summary}
This paper proposes a dynamical systems framework for realtime probabilistic monitoring of watershed hydrologic-hydrodynamic dynamics under partial measurements. In contrast to existing ensemble based uncertainty methods in the watershed literature~\cite{Moges2020,Nguyen2024a,Gomes2025,Castro2025}, the proposed approach computes confidence intervals for measured and unmeasured states from a single nominal simulation by using a closed-form covariance mapping that is distribution-agnostic, requiring only means and variances of uncertain inputs and parameters rather than a prescribed distributional form. The coupled overland flow routing and Green-Ampt infiltration dynamics are formulated as a semi-explicit nonlinear DAE, which provides a state space representation suitable for uncertainty propagation and estimation analysis. Based on this representation, the covariance is computed through the stacked linearized DAE and measurement model and evaluated using the closed-form mapping in Proposition~\ref{thm:pse_solution}. The case studies on V-Tilted and Walnut Gulch show that the proposed method remains consistent with baseline solvers and MC uncertainty references, while incorporating sensing information under partial gauging to reduce covariance at measured locations. Furthermore, the method provides computationally tractable uncertainty estimates in large scale domains, thereby enabling realtime confidence aware monitoring for operational forecasting and decision support in flood risk management and water resources applications. The paper is not without limitations and merits future work. The limitations and future work are summarized as follows.
\begin{itemize}[leftmargin=*]
\item The current framework considers parameter uncertainty but does not address structural model uncertainty, such as the choice of infiltration model or routing scheme. \comment{In particular, the Green-Ampt infiltration model adopted in~\eqref{eq:green_ampt_capacity}-\eqref{eq:infiltration_rate} approximates Richards' equation~\cite{richards1931capillary} and does not capture layered soil profiles, macropore and preferential flow, wetting and drying hysteresis, soil crusting~\cite{Becker2018}, or lateral subsurface flow; combined with the spatially uniform initial moisture $\theta_i$ used in the Walnut Gulch case study, these are additional uncertainty sources not captured by the propagated intervals. Similarly, the current framework adopts the diffusive wave approximation through the algebraic constraint $\m g(\m x_d,\m x_a)=\m 0$ in Section~\ref{subsec:dae}; extension to include inertial terms such as the local inertia approximation of~\cite{Gomes2023,Cozzolino2019} or the full dynamic wave requires rederiving the state space DAE; the covariance propagation in Algorithm~\ref{algorithm1} carries over to the reformulated system.} Future work can investigate approaches that account for model structure uncertainty\comment{, along with spatially variable $\theta_i$ in the watershed conditions}.\vspace{-0.3cm}
\item \comment{Assumption~\ref{asmp:independence} is a practical simplification. The covariance mapping in~\eqref{eq:Kxx} admits a spatial covariance matrix on $R$, $n$, $K_s$, $\psi_f$, and $\eta$ through the off diagonal blocks of $\mr{K}_{\theta}(k)$; quantifying how spatial correlation in $K_s$ and rainfall reshapes the propagated confidence intervals, jointly with a spatially variable $\theta_i$, is a direction for future work. Furthermore, the covariance mapping propagates the means and variances of the uncertain inputs through the linearized DAE; the skewness and excess kurtosis of the output distribution at the peak of extreme events are therefore not produced by the dynamics themselves, and any asymmetry in the confidence intervals is inherited from the adopted distribution model in Tab.~\ref{tab:pse_uncertainty}.}\vspace{-0.3cm}
\item The computational cost of assembling and inverting the system matrix $\m A[k]$ grows with the number of cells\comment{; for different grid resolutions or larger basin applications, the proposed framework may require memory that exceeds practical limits for the accumulated covariance blocks, and capturing spatial coupling through a global sparse assembly of $\m A[k]$ would increase the memory requirements}. Reduced order modeling techniques\comment{, such as Karhunen-Lo\`eve deep learning surrogates~\cite{Wang2025b},} or domain decomposition strategies\comment{~\cite{Her2021}, along with data driven surrogate approaches~\cite{Nguyen2024a} for approximating the Jacobian or accelerating the covariance computation,} can be explored to extend the applicability to larger watersheds.\vspace{-0.3cm}
\item The integration of the covariance framework with realtime data assimilation is not explored in this paper. Future research can investigate how to assimilate in-situ observations to update the state estimates and reduce uncertainty over time; the problem of optimally allocating sensing resources to minimize uncertainty in watershed conditions can be solved using established control theoretic frameworks as in~\cite{Kazma2024e}. The proposed framework is validated under event scale simulations with fixed spatial discretization; extension to adaptive mesh refinement, longer continuous simulations with spatially varying rainfall can be explored in future work.
\end{itemize}

Future research on this topic can investigate one or more of the above limitations of this paper.\\

\appendix
\section{Main Derivations}\label{apndx:proofs}
\subsection{Index of the proposed DAE}\label{apndx:index_dae}
\begin{proof}[Proof of Proposition~\ref{thm:index1_local}]
From~\eqref{eq:algebraic_constraint} and for a fixed terrain and roughness parameters $(\m z,\m n)$, we define the stacked directional Manning operator $\m \phi(\m h,\m z,\m n) := [\m \phi_{\mr{L}}^\top,\m \phi_{\mr{R}}^\top,\m \phi_{\mr{U}}^\top,\m \phi_{\mr{D}}^\top]^\top \in \Rn{n_a}$. Then, the stacked algebraic constraint consistent with~\eqref{eq:semi_NDAE_rep} can be written as
\[
\m g(\m x_d,\m x_a) = \m x_a - \m \phi(\m h,\m z,\m n) = \m 0,
\]
where $\m x_a = [\m Q_{\mr{L}}^\top,\m Q_{\mr{R}}^\top,\m Q_{\mr{U}}^\top,\m Q_{\mr{D}}^\top]^\top$ follows from~\eqref{eq:state_partition}. The Jacobian with respect to the algebraic states $\m x_a$ can be computed as $\frac{\partial \m g}{\partial \m x_a} = \m I_{n_a}$, which is nonsingular; this is consistent with Assumption~\ref{asmp:dae_regular}. Then, by the implicit function theorem, the algebraic states are locally and uniquely determined. Differentiating $\m g = \m 0$ once with respect to time eliminates the algebraic variables in the semi-explicit form as follows
\[
\dot{\m x}_a = \frac{\partial \m \phi}{\partial \m h} \dot{\m h}.
\]
This relation gives an explicit differential equation for computing the algebraic state variables. Therefore, one differentiation is sufficient to obtain an ODE representation for all state components, and the differentiation index is one~\cite{Volker2005,Linh2009}. Moreover, nonsingularity of $\partial \m g / \partial \m x_a$ implies local regularity of the linearized descriptor system around regular operating points~\cite{Volker2005}. This concludes the proof.
\end{proof}
\vspace{-0.1cm}
\subsection{Directional Manning Algebraic Map}\label{apndx:manning_map}
For completeness, the directional Manning CA algebraic map used in the DAE formulation can be written accordingly. For each cell $(i,j)$, define the water surface elevation as follows
\begin{equation}
\zeta^{i,j}=z^{i,j}+h^{i,j},
\end{equation}
then, for each direction $d\in\mc{D}_{\mr{dir}}$,
\begin{equation}
S_d^{i,j}=\max\left(\frac{\zeta^{i,j}-\zeta_d^{i,j}}{\Delta s_d},\,0\right),
\end{equation}
where $S_d^{i,j}$ is the water surface slope toward neighbor $d$ [m m$^{-1}$], $\zeta_d^{i,j}$ denotes the neighboring water surface elevation in direction $d$ [m], and $\Delta s_d$ is the directional cell distance [m]. If the neighboring cell is unavailable at the boundary, $\zeta_d^{i,j}$ is extrapolated using a free drainage boundary slope. Let
\begin{equation}
S_{\max}^{i,j}=\max_{d\in\mc{D}_{\mr{dir}}} S_d^{i,j},\qquad
S_{\mr{tot}}^{i,j}=\sum_{d\in\mc{D}_{\mr{dir}}}S_d^{i,j}.
\end{equation}
Under the admissible flow directions introduced in Section~\ref{subsubsec:manning}, only neighbors with nonnegative head difference are retained. Thus, $\Delta\eta_d^{i,j}=\zeta^{i,j}-\zeta_d^{i,j}\ge 0$, and the directional partition is written as
\[
w_d^{i,j}:=\frac{S_d^{i,j}}{S_{\mr{tot}}^{i,j}},\qquad S_{\mr{tot}}^{i,j}>0,\qquad
\sum_{d\in\mc{D}_{\mr{dir}}}w_d^{i,j}=1,
\]
which is consistent with~\eqref{eq:manning}.
Hence, for $S_{\mr{tot}}^{i,j}>0$, the directional discharges are
\begin{equation}
Q_d^{i,j}=\mathrm{ff}^{i,j}\,v^{i,j}\,h_{\mr{eff}}^{i,j}\,\ell_d\,
\frac{S_d^{i,j}}{S_{\mr{tot}}^{i,j}},\qquad d\in\mc{D}_{\mr{dir}},
\end{equation}
where the regularized terms $h_{\mr{reg}}^{i,j}$, $h_{\mr{eff}}^{i,j}$, $\mathrm{ff}^{i,j}$, and $v^{i,j}$ follow~\ref{apndx:regularization_terms}. Note that if $S_{\mr{tot}}^{i,j}=0$, then $Q_d^{i,j}=0$ for all $d\in\mc{D}_{\mr{dir}}$. The corresponding algebraic constraints in the DAE are
\begin{equation}
0=Q_d^{i,j}-\phi_d^{i,j}\left(h^{i,j},h_d^{i,j},z^{i,j},z_d^{i,j},n^{i,j}\right),\qquad d\in\mc{D}_{\mr{dir}}.
\end{equation}
Here, $\phi_d^{i,j}(\cdot)$ denotes the cellwise component of the global mapping $\m\phi_d(\m h,\m z,\m n)$ in~\eqref{eq:algebraic_constraint}.

\subsection{Proof of the closed-form covariance proposition}\label{apndx:kxx_solution}
\begin{proof}[Proof of Proposition~\ref{thm:pse_solution}]
From~\eqref{eq:cov_relation_linear}, the covariance constraint in Problem~\ref{prob:pse_cov_linear_form} is written as
\[
\m A[k]\,\mr{K}_{\m x\m x}[k]\,\m A[k]^\top=\mr{K}_{\m b\m b}[k].
\]
Under Assumption~\ref{asmp:dae_regular}, $\m A[k]$ has full column rank, so $\m A[k]^\top\m A[k]$ is nonsingular and
\[
(\m A[k])^{\dagger}=\left(\m A[k]^\top\m A[k]\right)^{-1}\m A[k]^\top.
\]
Premultiplying by $(\m A[k])^{\dagger}$ and postmultiplying by $((\m A[k])^{\dagger})^\top$ gives
\[
\mr{K}_{\m x\m x}[k]=(\m A[k])^{\dagger}\mr{K}_{\m b\m b}[k]((\m A[k])^{\dagger})^\top,
\]
which is equivalent to~\eqref{eq:Kxx}. This concludes the proof.
\end{proof}
\vspace{-0.1cm}

\section{Regularization Terms}\label{apndx:regularization_terms}
The DAE formulation uses smooth regularization terms to improve numerical robustness on steep terrain. For compactness, cell indices are omitted, and each expression is applied cellwise.
\begin{equation}
h_{\mr{reg}}=\sqrt{h^2+h_{\mr{smooth}}^2},\qquad h_{\mr{eff}}=\max(h_{\mr{reg}}-h_0,\;h_{\mr{smooth}}),
\end{equation}
where $h_{\mr{smooth}}>0$ is a small smoothing depth and $h_0\ge 0$ is the ponding threshold.
\begin{equation}
\mathrm{ff}(h)=\xi^2(3-2\xi),\qquad \xi=\min\left(1,\frac{h}{h_{\mr{flow}}}\right),
\end{equation}
where $h_{\mr{flow}}>0$ is the flow activation depth.
\begin{equation}
v_{\mr{Manning}}=\frac{1}{n}h_{\mr{eff}}^{2/3}\sqrt{S_{\max}},\qquad v=\min(v_{\mr{Manning}},v_{\max}).
\end{equation}
Here, $S_{\max}\ge 0$ is the local maximum water surface slope and $v_{\max}>0$ is a velocity cap used for numerical robustness.

For Green-Ampt infiltration, a smooth minimum is used between capacity and available supply. In the BDF formulation, the supply term is evaluated using previous step ponding depth and is written as
\begin{equation}
f_{\mr{cap}}=K_s\left(1+\frac{\Delta\theta\,(h+\psi_f)}{\max(F,F_{\min})}\right),\qquad
i_a=\frac{h_{k}}{\Delta t}+R_{k+1},
\end{equation}
\begin{equation}
f=\frac{1}{2}\left(f_{\mr{cap}}+i_a-\sqrt{(f_{\mr{cap}}-i_a)^2+\varepsilon_s}\right).
\end{equation}
The parameter $\varepsilon_s>0$ controls the smooth-min regularization. Here, $\Delta\theta$ is the Green-Ampt moisture deficit term, denoted by $\eta$ in Section~\ref{subsubsec:uncertainty_sources}. This smoothing avoids discontinuities of a hard \textit{min} operator and improves convergence of the nonlinear Newton-Raphson iteration.

\balance
\vspace{0.1cm}
 \bibliographystyle{elsarticle-harv}
\vspace{0.1cm}
\bibliography{library_awr2026}
\end{document}